%% file: main.tex
\documentclass[lettersize,journal]{IEEEtran}
\input{pretex}
\usepackage{amsmath,amsfonts,amssymb,array,graphicx,mathtools,multirow,bm,tcolorbox,relsize,booktabs}
\usepackage{algorithmic}
\usepackage{array}
\usepackage[caption=false,font=normalsize,labelfont=sf,textfont=sf]{subfig}
\usepackage{textcomp}
\usepackage{stfloats}
\usepackage{url}
\usepackage{verbatim}
\usepackage{graphicx}
\def\BibTeX{{\rm B\kern-.05em{\sc i\kern-.025em b}\kern-.08em
    T\kern-.1667em\lower.7ex\hbox{E}\kern-.125emX}}
\usepackage{balance}
\usepackage{gensymb}
\usepackage{amsfonts}
\usepackage{cite}
\usepackage{caption}

\begin{document}
\title{Estimate distillable entanglement and quantum capacity by squeezing useless entanglement}
\author{Chengkai Zhu
\IEEEmembership{Student Member, IEEE}, Chenghong Zhu
\IEEEmembership{Student Member, IEEE}, Xin Wang
\IEEEmembership{Member, IEEE}
\thanks{
This work was supported by the Start-up Fund (No. G0101000151) from The Hong Kong University of Science and Technology (Guangzhou), the Guangdong Provincial Quantum Science Strategic Initiative (No. GDZX2303007), and the Education Bureau of Guangzhou Municipality.
Part of this work was done when the authors were at Baidu Research. C.Z and C.Z contributed equally to this work. (\textit{Corresponding author: Xin Wang.})

The authors are with the Thrust of Artificial Intelligence, Information Hub, The Hong Kong University of Science and Technology (Guangzhou), Guangzhou 511453, China. (email: czhu696@connect.hkust-gz.edu.cn; czhu854@connect.hkust-gz.edu.cn; felixxinwang@hkust-gz.edu.cn)}
}

\maketitle

\begin{abstract}
Quantum Internet relies on quantum entanglement as a fundamental resource for secure and efficient quantum communication, reshaping data transmission. In this context, entanglement distillation emerges as a crucial process that plays a pivotal role in realizing the full potential of the quantum internet. Nevertheless, it remains challenging to accurately estimate the distillable entanglement and its closely related essential quantity, the quantum capacity. In this work, we consider a general resource measure known as the reverse divergence of resources which quantifies the minimum divergence between a target state and the set of free states. Leveraging this measure, we propose efficiently computable upper bounds for both quantities based on the idea that the useless entanglement within a state or a quantum channel does not contribute to the distillable entanglement or the quantum capacity, respectively. Our bounds can be computed via semidefinite programming and have practical applications for purifying maximally entangled states under practical noises, such as depolarizing and amplitude damping noises, leading to improvements in estimating the one-way distillable entanglement. Furthermore, we provide valuable benchmarks for evaluating the quantum capacities of qubit quantum channels, including the Pauli channels and the random mixed unitary channels, which are of great interest for the development of a quantum internet.
\end{abstract}


\section{Introduction}
\label{sec:introduction}
\subsection{Background}

\IEEEPARstart{Q}{uantum} Internet~\cite{Kimble2008,Cacciapuoti2020,Wehner2018b,Illiano2022}, a transformative paradigm in data transmission, leverages the extraordinary properties of quantum entanglement to revolutionize communication. At the heart of this groundbreaking innovation lies quantum entanglement, which is the most nonclassical manifestation of quantum mechanics and serves as the cornerstone resource for quantum information processing~\cite{Bennett1993,Bennett1992,Bennett1984,Ekert1991}. Simultaneously, quantum communication emerges as a unique task, utilizing the power of entanglement for secure and efficient information transfer~\cite{Bennett1984}. To fully unlock the potential of quantum technologies and achieve a secure, efficient, and globally connected quantum internet, it is desirable to push the boundaries of our understanding and capabilities in entanglement and quantum communication theories.

In the context of entanglement theory, maximally entangled states are often used as a resource for various quantum protocols and algorithms. However, noise inevitably occurs in practical scenarios, resulting in mixed entangled states that require distillation or purification. A natural question is how to obtain the maximally entangled states from a source of less entangled states using well-motivated operations, known as \textit{entanglement distillation}. One fundamental measure for characterizing the entanglement distillation is the \textit{one-way distillable entanglement}~\cite{Devetak2003a}, denoted by $E_{D,\to}$, which is also one of the most important entanglement measures motivated by operational tasks. It captures the highest rate at which one can obtain maximally entangled states from less entangled states by one-way local operations and classical communication (LOCC):
\begin{equation*}
E_{D,\to}(\rho_{AB})=\sup\{r:\lim_{n \to \infty} [\inf_{\Lambda}  \|\Lambda(\rho_{AB}^{\ox n})- \Phi(2^{rn})\|_1]=0\},
\end{equation*}
where $\Lambda$ ranges over one-way LOCC operations and $\Phi(d)=1/d\sum_{i,j=1}^d \ketbra{ii}{jj}$ is the standard $d\otimes d$ maximally entangled state. Likewise, the \textit{two-way distillable entanglement} $E_{D,\leftrightarrow}(\rho_{AB})$ is defined by the supremum over all achievable rates under two-way LOCC. We have for all bipartite states $\rho_{AB}$ that $E_{D,\to}(\rho_{AB}) \leq E_{D,\leftrightarrow}(\rho_{AB})$. Notably, the distillable entanglement is closely connected to the fundamental notion of \textit{quantum capacity} in quantum communication~\cite{singh2022fully}, which is central to quantum Shannon theory. Consider modeling the noise in transmitting quantum information from Alice to Bob as a quantum channel $\cN_{A\to B}$. The quantum capacity $Q(\cN_{A\to B})$ is the maximal achievable rate at which Alice can reliably transmit quantum information to Bob by asymptotically many uses of the channel. By the state–channel duality, if the distillation protocol of the Choi state~\cite{Choi1975} $J_{AB}^\cN$ of $\cN_{A\to B}$ yields the maximally entangled states at a positive rate, then Alice may apply the standard teleportation scheme to send arbitrary quantum states to Bob at the same rate. Thus, one has $Q(\cN_{A\to B}) \geq E_{D,\to}(J_{AB}^\cN)$ since classical forward communication in teleportation does not affect the channel capacity. For the teleportation-simulable channels~\cite{BK98,Bennett1993,Werner2001a}, the equality here holds.

Despite many efforts that have been made in the past two decades, computing $E_{D,\to}(\cdot)$ and $Q(\cdot)$ still generally remains a challenging task. Even for the qubit isotropic states and the depolarizing channels, it remains unsolved. Therefore, numerous studies try to estimate them by deriving lower and upper bounds (see, e.g.,~\cite{Devetak2003a,Rains2001, Wang2016,Hayashi2006a,Leditzky2017,Kaur2018,Wang2019b} for the distillable entanglement, e.g.,~\cite{Ouyang2014, Kianvash2022, Wolf2007, Fanizza2019} for the quantum capacity). For the distillable entanglement, a well-known lower bound dubbed \textit{Hashing bound} was established by Devetak and Winter~\cite{Devetak2003a}, i.e., $E_{D,\to}(\rho_{AB}) \geq I_{\textup{\tiny C}}(A\rangle B)_{\rho}$, where $I_{\textup{\tiny C}}(A\rangle B)_{\rho}$ is the coherent information of the bipartite quantum state $\rho_{AB}$. Considering upper bounds, the Rains bound~\cite{Rains2001} is arguably the best-known efficiently computable bound for the two-way distillable entanglement of general states. Recent works~\cite{Leditzky2017,Wang2019b} utilize the techniques that involve constructing meaningful extended states to find upper bounds. For quantum capacity, many useful upper bounds for general quantum channels are studied for benchmarking arbitrary quantum noise~\cite{Holevo2001, Sutter2014, Muller-Hermes2015, Wang2016a, Wang2017d, Pisarczyk2018, Pirandola2015b,Fang2019a}. When considering some specific classes of quantum channels, useful upper bounds are also developed to help us better understand quantum communication via these channels~\cite{Cerf2000, Wolf2007, Smith2008a, Gao2015a, Wang2019b, Fanizza2019, Kianvash2022}.

In specific, due to the regularization in the characterizations of $E_{D,\to}(\cdot)$ and $Q(\cdot)$, one main strategy to establish efficiently computable upper bounds on them is to develop single-letter formulae. For example, one common approach is to decompose a state (resp. a quantum channel) into degradable parts and anti-degradable parts~\cite{Wolf2007}, or use approximate degradability (anti-degradability)~\cite{Sutter2014}. Another recent fruitful technique called flag extension optimization~\cite{Fanizza2019, Wang2019b, Kianvash2022} relies on finding a degradable extension of the state or the quantum channel. However, the performance of these methods is limited by the absence of a good decomposition strategy. It is unknown how to partition a general state or quantum channel to add flags or how to construct a proper and meaningful convex decomposition on them. Thus, the flag extension optimization is only effective for the states and channels with symmetry or known structures.

\subsection{Main contributions}
This work considers a family of resource measures called \textit{reverse divergence of resources}. With a specific construction, we define the \textit{reverse max-relative entropy of entanglement} for quantum states, which is inspired by the idea to decompose a state into a degradable part and anti-degradable part raised in~\cite{Leditzky2017}. It has applications for estimating the distillable entanglement. In the meantime, we introduce the \textit{reverse max-relative entropy of anti-degradability} for quantum channels as a generalization of the concept of that for states, which can be applied to bound the quantum capacity. All these bounds can be efficiently computed via semidefinite programming~\cite{Vandenberghe1996}. Furthermore, drawing on the idea of~\cite{Sutter2014}, we thoroughly analyze different continuity bounds on the one-way distillable entanglement of a state in terms of its anti-degradability. Finally, we investigate the distillation of the maximally entangled states under practical noises and focus on the quantum capacity of qubit channels. We show that the bound obtained by the reverse max-relative entropy of entanglement outperforms other known bounds in a high-noise region, including the Rains bound and the above continuity bounds. The upper bound offered by the reverse max-relative entropy of anti-degradability also provides an alternative interpretation of the no-cloning bound of the Pauli channel~\cite{Cerf2000}, and notably outperforms the continuity bounds on random unital qubit channels.

This paper is structured as follows. Section \ref{sec:preliminary} provides preliminaries used throughout the paper. Section \ref{sec:rev_div} introduces a family of resource measures called the reverse divergence of resources. Section \ref{sec:app_on_DE} shows the application of this concept on upper bounding the distillable entanglement, with practical examples provided. Section \ref{sec:app_on_QC} demonstrates the application of the method in deriving upper bounds on quantum capacity, including analytical results for Pauli channels and comparison with continuity bounds in subsection~\ref{sec:qc_pauli_chan} for random mixed unitary channels. The paper concludes with a summary and outlooks for future research in section \ref{sec:conclusion}.

\section{Reverse divergence of resources} \label{sec:pre_and_rev_div}
\subsection{Preliminaries}\label{sec:preliminary}
Let $\cH$ be a finite-dimensional Hilbert space, and $\cL(\cH)$ be the set of linear operators acting on it. We consider two parties Alice and Bob with Hilbert space $\mathcal{H}_A, \mathcal{H}_B$, whose dimensions are $d_A, d_B$, respectively. A linear operator $\rho \in \cL(\cH)$ is called a density operator if it is Hermitian and positive semidefinite with trace one. We denote the trace norm of $\rho$ as $\|\rho\|_1 = \tr \sqrt{\rho^\dagger \rho}$ and let $\cD(\cH)$ denote the set of density operators. We call a linear map CPTP if it is both completely positive and trace-preserving. A CPTP map that transforms linear operators in $\cL(\cH_A)$ to linear operators in $\cL(\cH_B)$ is also called a quantum channel, denoted as $\cN_{A\to B}$. For a quantum channel $\cN_{A\to B}$, its Choi-Jamiołkowski state is given by $J_{AB}^{\cN} \equiv \sum_{i, j=0}^{d_A-1}\ketbra{i}{j} \ox \mathcal{N}_{A \to B}(\ketbra{i}{j})$, where $\{|i\rangle\}_{i=0}^{d_A-1}$ is an orthonormal basis in $\cH_A$. The von Neumann entropy of a state $\rho_A$ is $S(A)_\rho := - \tr(\rho_A \log \rho_A)$ and the coherent information of a bipartite state $\rho_{AB}$ is defined by $I_c(A\rangle B)_\rho := S(B)_\rho - S(AB)_\rho$. The entanglement of formation of a state $\rho_{AB}$ is given by
\begin{equation}
E_F(\rho_{AB}) = \min_{\{p_i, \ket{\phi_i}\}} \sum_{i} p_i S(A)_{\phi_i},    
\end{equation}
where $\rho_{AB}=\sum_i p_i \ketbra{\phi_i}{\phi_i}_{AB}$ and the minimization ranges over all pure state decomposition of $\rho_{AB}$. We introduce the generalized divergence $\boldsymbol{D}(\rho_A \| \sigma_A)$ as a map $\boldsymbol{D}: \cD(\cH_A)\times \cD(\cH_A) \mapsto \mathbb{R}\cup \{+\infty\}$ that obeys:
\begin{enumerate}
    \item Faithfulness: $\boldsymbol{D}(\rho_A \| \sigma_A)=0$ iff $\rho_A = \sigma_A$.
    \item Data processing inequality: $\boldsymbol{D}(\rho_A \| \sigma_A) \geq \boldsymbol{D}[\mathcal{N}_{A\to A'}(\rho_A) \| \mathcal{N}_{A\to A'}(\sigma_A)]$, where $\mathcal{N}_{A\to A'}$ is an arbitrary quantum channel.
\end{enumerate}
The generalized divergence is intuitively some measure of distinguishability of the two states, e.g., Bures metric, quantum relative entropy. Another example of interest is the \textit{sandwiched R{\'{e}}nyi relative entropy}~\cite{Wilde_2014, Muller_Lennert2013} of $\rho, \sigma$ that is defined by
\begin{equation}
    D_\alpha(\rho \| \sigma):=\frac{1}{\alpha-1} \log \operatorname{Tr}\left[\left(\sigma^{\frac{1-\alpha}{2 \alpha}} \rho \sigma^{\frac{1-\alpha}{2 \alpha}}\right)^\alpha\right],
\end{equation}
if $\operatorname{supp}(\rho) \subset \operatorname{supp}(\sigma)$ and it is equal to $+\infty$ otherwise, where $\alpha \in(0,1) \cup(1, \infty)$. In the case that $\alpha \to \infty$, one can find the \textit{max-relative entropy}~\cite{datta2009min} of $\rho$ with respect to $\sigma$ by
\begin{equation}
    D_{\max}(\rho || \sigma) = \inf \{\lambda \in \mathbb{R}: \rho \leq 2^{\lambda} \sigma\}.
\end{equation}

\subsection{Reverse divergence of resources}\label{sec:rev_div}
In the usual framework of quantum resource theories~\cite{chitambar2019quantum}, there are two main notions: i) subset $\cF$ of free states, i.e., the states that do not possess the given resource; ii) subset $\cO$ of free operations, i.e., the quantum channels that are unable to generate the resource. Meanwhile, two axioms for a quantity being a resource measure $\cR(\cdot)$ are essential:
\begin{itemize}
    \item[1).] Vanishing for free states: $\rho\in \cF \Rightarrow \cR(\rho) = 0$.
    \item[2).] Monotonicity: $\cR(\cO(\rho)) \leq \cR(\rho)$ for any free operation $\cO$. $\cR(\cdot)$ is called a \textit{resource monotone}.
\end{itemize}
Let us define a family of resource measures called \textit{reverse divergence of resources}:
\begin{equation}\label{Eq:rev_diver_resou}
    \cR_{\cF}(\rho) := \min_{\tau\in \cF}\boldsymbol{D}(\tau||\rho),
\end{equation}
where $\mathcal{F}$ is some set of free states. By the definition of the reverse divergence of resources in Eq.~\eqref{Eq:rev_diver_resou}, one can easily check it satisfies condition 1). Whenever the free state set $\cF$ is closed under the free operations, by the data-processing inequality of $\boldsymbol{D}(\cdot\| \cdot)$, condition 2) will be satisfied. Thus $\cR_{\mathcal{F}}(\cdot)$ is a resource measure. Specifically, in the resource theory of entanglement, some representative free state sets are the separable states (SEP) and the states having a positive partial transpose (PPT). Examples of free sets of operations are LOCC and PPT. We note that the "reverse" here means minimizing the divergence over a free state set in the first coordinate, rather than the second one which has helped one define the \textit{relative entropy of entanglement}~\cite{Vedral_1997quantify, Vedral_1997} and the \textit{max-relative entropy of entanglement}~\cite{datta2009min}. For some divergence of particular interest, e.g., the quantum relative entropy $D(\cdot||\cdot)$, relevant discussion of the coordinate choices can be traced back to work in~\cite{Vedral_1997}. In \cite{Eisert_2003}, the authors further studied properties of the quantity $\min_{\tau\in \cF}D(\tau||\rho)$. Here, we try to further investigate meaningful applications of some reverse divergence of resources.

In the following, we consider the generalized divergence as the max-relative entropy and study a measure called \textit{reverse max-relative entropy of resources},
\begin{equation}\label{Eq:sqz_mrela_def}
    \cR_{\max,\mathcal{F}}(\rho) := \min_{\tau\in \mathcal{F}}D_{\max}(\tau||\rho),
\end{equation}
where $\mathcal{F}$ is some set of free states. If there is no such state $\tau\in \cF$ that satisfies $\tau \leq 2^{\lambda} \rho$ for any $\lambda \in \mathbb{R}$, $\cR_{\max,\mathcal{F}}(\rho)$ is set to be 0. 
This measure bears many nice properties. First, it can be efficiently computed via semidefinite programming (SDP) in many cases. Second, Eq.~\eqref{Eq:sqz_mrela_def} gives the closest free state $\tau\in\cF$ to $\rho$, w.r.t. the max-relative entropy. Third, $\cR_{\max,\mathcal{F}}(\cdot)$ is subadditive w.r.t the tensor product of states. In fact, $\cR_{\max,\mathcal{F}}(\rho)$ is closely related to the \textit{weight of resource} $W(\rho)$~\cite{ELITZUR199225,Lewenstein98,Ducuara2020,Uola2020} and the \textit{free component} $\Gamma(\rho)$~\cite{Fang2020}, both of which have fruitful properties and applications~\cite{Regula2021,regula2022overcoming,Regula_2023a}, as follows
\begin{equation}\label{Eq:WoR_fcomp}
    2^{-\cR_{\max,\mathcal{F}}(\rho)}  = 1-W(\rho) = \Gamma(\rho).
\end{equation}
We note that each part of Eq.~\eqref{Eq:WoR_fcomp} quantifies the largest weight that a free state can take in a convex decomposition of $\rho$. When moving on to operational tasks that the free state can be ignored, what is left in a convex decomposition becomes our main concern. Optimization of the weight in the decomposition can be visualized as \textit{squeezing out} all free parts of the given state. Thus, we further introduce the \textit{$\cF$-squeezed state} of $\rho$ as follows.
\begin{definition}\label{def:sqz_state}
For a quantum state $\rho$ and a free state set $\cF$, if $\cR_{\max,\mathcal{F}}(\rho)$ is non-zero, the $\cF$-squeezed state of $\rho$ is defined by
\begin{equation}
    \omega = \frac{\rho-2^{-\cR_{\max,\mathcal{F}}(\rho)}\cdot\tau}{1-2^{-\cR_{\max,\mathcal{F}}(\rho)}},
\end{equation}
where $\tau$ is the closest free state to $\rho$ in terms of the max-relative entropy, i.e., the optimal solution in Eq.~\eqref{Eq:sqz_mrela_def}. If $\cR_{\max,\mathcal{F}}(\rho)=0$, the $\cF$-squeezed state of $\rho$ is itself.
\end{definition}

In the following sections, we will illustrate the applications of the reverse max-relative entropy of resources as well as the squeezing idea on example tasks. One is giving upper bounds on the distillable entanglement of arbitrary quantum states. The other is to derive upper bounds on the quantum capacity of channels. 

\section{Applications on distillable entanglement} \label{sec:app_on_DE}
In this section, we investigate the information-theoretic application of the reverse max-relative entropy of resources in deriving efficiently computable upper bounds on the distillable entanglement. To showcase the advantage of our bounds, we compare the results with different continuity bounds and the Rains bound on the maximally entangled states with practical noises.

\subsection{Upper bound on the one-way distillable entanglement}\label{sec:one_way_app}
Recall that the one-way distillable entanglement has a regularized formula~\cite{Devetak2003a}:
\begin{align}\label{Eq:reg_form_de}
    E_{D,\to}(\rho_{AB}) = \lim_{n\rightarrow \infty} \frac{1}{n}E^{(1)}_{D,\to}(\rho_{AB}^{\otimes n}),
\end{align}
where 
$E^{(1)}_{D,\to}(\rho_{AB}) := \max_{T} I_{\textup{\tiny C}}(A'\rangle M B)_{T(\rho_{AB})}$,
and the maximization ranges over all quantum instruments
$T:A\to A'M$ on Alice’s system. The regularization in Eq.~\eqref{Eq:reg_form_de} for $E_{D,\to}(\rho_{AB})$ is intractable to compute in most cases. However, there are some categories of states whose $E_{D,\to}$ can be reduced to single-letter formulae. Two important classes are called \textit{degradable states} and \textit{anti-degradable states}.

Let $\rho_{AB}$ be a bipartite state with purification $\ket{\phi}_{ABE}$. $\rho_{AB}$ is called degradable if there exists a CPTP map $\mathcal{M}_{B\rightarrow E}$ such that $\mathcal{M}_{B\rightarrow E}(\rho_{AB}) = \tr_B(\phi_{ABE})$, and is called anti-degradable if there exists a CPTP map $\mathcal{M}_{E\rightarrow B}$ such that $\mathcal{M}_{E\rightarrow B}(\rho_{AE}) = \tr_E(\phi_{ABE})$. Equivalently, a state is anti-degradable if and only if it has a symmetric extension~\cite{myhr2011symmetric}, thus is also called a \textit{two-extendible state}. For the degradable states, it is shown that~\cite{Leditzky2017}
\begin{align}\label{Eq:deg_E1}
     E^{(1)}_{D,\to}(\rho_{AB}^{\otimes n}) = n  E^{(1)}_{D,\to}(\rho_{AB}) = n I_{\textup{\tiny C}}(A\rangle B)_{\rho}, \forall n\in \mathbb{N},
\end{align}
resulting in $E_{D,\to}(\rho_{AB}) =  I_{\textup{\tiny C}}(A\rangle B)_{\rho}$. For the anti-degradable states, consisting of a compact and convex set, it always holds
\begin{equation}\label{Eq:adg_E1}
    E^{(1)}_{D,\to}(\rho_{AB}) = E_{D,\to}(\rho_{AB}) = 0.
\end{equation}
Moreover, Leditzky, Datta, and Smith~\cite{Leditzky2017} showed that $E_{D,\to}(\cdot)$ is convex on decomposing a state into degradable and anti-degradable parts, i.e.,
\begin{equation}\label{Eq:deg_adg_conv}
    E^{(1)}_{D,\to}(\rho_{AB}) \leq s\cdot E^{(1)}_{D,\to}(\omega_{AB}) + (1-s)\cdot E^{(1)}_{D,\to}(\tau_{AB}),
\end{equation}
where $\rho_{AB} = s\omega_{AB} + (1-s)\tau_{AB}$, $\omega_{AB}$ is degradable, and $\tau_{AB}$ is anti-degradable.

To better exploit this convexity, we introduce the \textit{reverse max-relative entropy of unextendible entanglement} to help identify the anti-degradable (two-extendible) part of a given bipartite state $\rho_{AB}$:
\begin{equation}\label{Eq:sqz_mradg_def}
    \cR_{\max,\adg}(\rho_{AB}) := \min_{\tau\in \adg}D_{\max}(\tau_{AB}||\rho_{AB}),
\end{equation}
where $\adg$ is the set of all anti-degradable states. In this resource theory, the resource we consider is the extendible entanglement, and the free states are bipartite states that are possibly shareable between $A$ and a third party $E$, where $E$ is isomorphic to $B$. Notably, the extendibility of entanglement is a key property in entanglement theory with many existing applications~\cite{Kaur2018,Wang2019h}. Here, combined with the idea of entanglement of formation, $\cR_{\max,\adg}(\rho_{AB})$ can be applied to derive an upper bound on the one-way distillable entanglement of an arbitrary state $\rho_{AB}$ as shown in Theorem~\ref{thm:anti_sqz_bound}.

\begin{figure}[t]
    \centering
    \includegraphics[width=\linewidth]{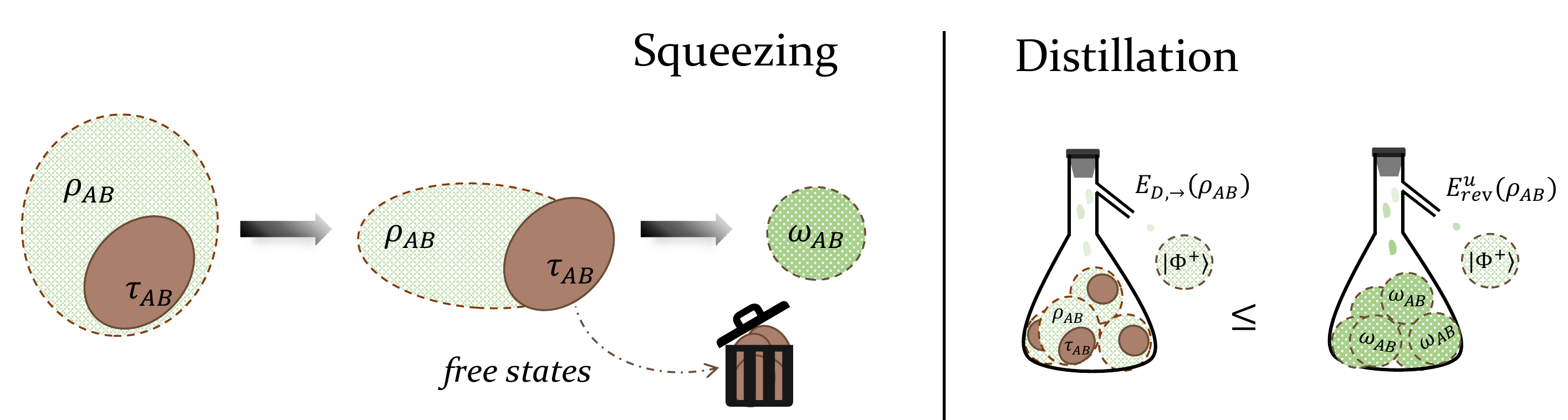}
    \caption{Illustration of the $\cF$-squeezed state of $\rho_{AB}$ and Theorem~\ref{thm:anti_sqz_bound}. The brown region corresponds to the sub-state $\tau_{AB}$ belonging to the free state set $\cF$, which is squeezed out via a convex decomposition of $\rho_{AB}$. $\omega_{AB}$ is the $\cF$-squeezed state.}
    \label{fig:cartoon_sqz_distill}
\end{figure}

\begin{theorem}\label{thm:anti_sqz_bound}
For any bipartite state $\rho_{AB}$,
\begin{equation}\label{Eq:oneway_bound}
\begin{aligned}
E_{D,\to}(\rho_{AB}) &\leq E_{\rm rev}^{u}(\rho_{AB})\\
&:=[1-2^{-\cR_{\max,\adg}(\rho_{AB})}]\cdot E_F(\omega_{AB}),    
\end{aligned}
\end{equation}
where $\omega_{AB}$ is the $\adg$-squeezed state of $\rho_{AB}$, $\cR_{\max,\adg}(\cdot)$ is the reverse max-relative entropy of unextendible entanglement, and $E_F(\cdot)$ is the entanglement of formation.
\end{theorem}
\begin{proof}
Suppose the ADG-squeezed state of $\rho_{AB}$ is $\omega_{AB}$ and the optimal solution in Eq.~\eqref{Eq:sqz_mrela_def} for $\rho_{AB}$ is $\tau_{AB}$. It follows $\rho_{AB} = [1-2^{-\cR_{\max,\adg}(\rho_{AB})}]\omega_{AB} + 2^{-\cR_{\max,\adg}(\rho_{AB})}\tau_{AB}$,
where $\tau_{AB}$ is anti-degradable. Since any pure state is degradable, we can decompose $\rho_{AB}$ into degradable and anti-degradable parts as
$\rho_{AB} = [1-2^{-\cR_{\max,\adg}(\rho_{AB})}]\sum_i p_i \ketbra{\omega_i}{\omega_i}_{AB} + 2^{-\cR_{\max,\adg}(\rho_{AB})}\tau_{AB}$, where $\omega_{AB}= \sum_{i}p_i \ketbra{\omega_i}{\omega_i}_{AB}$ is any pure state decomposition of $\omega_{AB}$. Based on the essential convexity of $E_{D,\to}(\cdot)$ on decomposing a state into degradable and anti-degradable parts~\cite{Leditzky2017} (cf. Eq.~\eqref{Eq:deg_adg_conv}), we have
\begin{equation*}
\begin{aligned}
    E_{D,\to}(\rho_{AB}) &\leq [1-2^{-\cR_{\max,\adg}(\rho_{AB})}]\cdot \sum_{i} p_i E_{D,\to}(\omega_{i})\\
    & \quad + 2^{-\cR_{\max,\adg}(\rho_{AB})} E_{D,\to}(\tau_{AB})\\
    &\leq [1-2^{-\cR_{\max,\adg}(\rho_{AB})}]\cdot \sum_{i} p_i I_c(A\rangle B)_{\omega_{i}},
\end{aligned}
\end{equation*}
where the second inequality is due to the properties in Eq.~\eqref{Eq:deg_E1} and Eq.~\eqref{Eq:adg_E1} for degradable states and anti-degradable states, respectively. After taking the minimization over all possible decomposition of $\omega_{AB}$, we arrived at $E_{D,\to}(\rho_{AB}) \leq [1-2^{-\cR_{\max,\adg}(\rho_{AB})}]\cdot E_F(\omega_{AB})$.
\end{proof}

\begin{remark}
The bound $E_{\rm rev}^{u}(\rho_{AB})$ has a cartoon illustration as shown in Fig.~\ref{fig:cartoon_sqz_distill}. The main insight of it is to squeeze out as much of the free or useless part, the anti-degradable state here, as possible. We point out that squeezing \textit{all} useless parts out does not necessarily give the tightest upper bound in terms of the one-way distillable entanglement, e.g., isotropic states~\cite{Leditzky2017}. Instead of squeezing out all the useless parts, there may be room for exploring more appropriate decompositions when we try to decompose a specific quantum state. However, the approach we present in Theorem~\ref{thm:anti_sqz_bound} is an invaluable method for general states as shown in subsection~\ref{sec:conti_bound_ex} and can be seen as a variant of the continuity bound in terms of the anti-degradability of the state.
\end{remark}
\begin{corollary}\label{coro:E_sqz_upperbound}
For any bipartite state $\rho_{AB}$,
\begin{equation*}\label{Eq:sqz_bound}
\begin{aligned}
    E_{D,\to}(\rho_{AB}) &\leq \Hat{E}_{\rm rev}^u(\rho_{AB})\\
    &:= [1-2^{-\cR_{\max,\adg}(\rho_{AB})}]\cdot \sum_{i} \lambda_i S(B)_{\psi_i},
\end{aligned}
\end{equation*}
where $\omega_{AB}=\sum_{i} \lambda_i \ketbra{\psi_i}{\psi_i}$ is the spectral decomposition of the ADG-squeezed state $\omega_{AB}$ of $\rho_{AB}$.
\end{corollary}
Corollary \ref{coro:E_sqz_upperbound} is followed by the fact that $E_F(\omega_{AB})$ has a trivial upper bound as $E_F(\omega_{AB}) \leq \sum_{i} \lambda_i S(B)_{\psi_i}$. We note that any other good upper bounds on the entanglement of formation can also be applied to Theorem~\ref{thm:anti_sqz_bound}. In particular, the bound $\Hat{E}_{\rm rev}^u(\rho_{AB})$ is efficiently computable since $\cR_{\max,\adg}(\rho_{AB})$ and $2^{-\cR_{\max,\adg}(\rho_{AB})}$ can be efficiently computed via an SDP. 
By Slater’s condition~\cite{giorgi2013traces}, the primal and dual problems satisfy strong duality, and both evaluate to $1-2^{-\cR_{\max,\adg}(\rho_{AB})}$.  We remain the derivation of the dual program in Appendix~\ref{appendix:sqzbound_dual_sdp}.
\begin{equation*}\label{Eq:pri_dual_SDP}
\begin{aligned}
&\underline{\textbf{Primal Program}}\\
\min_{\omega_{AB}, \tau_{AB}, \tau_{ABE}}&\;  \tr[\omega_{AB}],\\
 {\rm s.t.}& \; \rho_{AB} =\omega_{AB}+\tau_{AB}, \\
&\; \omega_{AB} \geq 0, \tau_{AB} \geq 0, \tau_{ABE} \geq 0,\\
&\; \tr_{E}[\tau_{ABE}] = \tr_{B}[\tau_{ABE}] = \tau_{AB}.
\end{aligned}
\end{equation*}

It is worth mentioning that this new bound is related to the bound $E_{\rm DA}(\cdot)$ proposed in~\cite{Leditzky2017} utilizing the convexity of $E_{D,\to}(\cdot)$ on decomposing a state into degradable and anti-degradable parts. Remarkably, constructing such a decomposition is challenging due to the non-convex nature of the degradable state set. Thus, it is difficult to compute $E_{\rm DA}(\cdot)$ in practice due to the hardness of tracking all possible decompositions. For example, when considering maximally entangled states affected by general noises, we currently lack strategies to compute the exact value of $E_{\rm DA}(\cdot)$. In contrast, $\Hat{E}_{\rm rev}^u(\cdot)$ overcomes this difficulty and can be efficiently computed for general states, serving as a valuable approximation of $E_{\rm DA}(\cdot)$. It outperforms the known computable upper bounds for many maximally entangled states under practical noises in a high-noise region, as shown in subsection \ref{sec:conti_bound_ex}. Moreover, our method in particular offers flexibility in selecting other decomposition strategies, i.e., the object functions in the SDP, beyond simply maximizing the proportion of antidegradable states. Such a method provides a new algorithmic route to fully exploit the bound given in~\cite{Leditzky2017}.

\subsection{Continuity bounds of the one-way distillable entanglement}\label{sec:conti_bound}
Note that the insight of the bound above is considering the distance of a given state between the anti-degradable states set. With different distance measures, the authors in~\cite{Sutter2014} derived continuity bounds on quantum capacity in terms of the (anti)degradability of the channel. Thus for self-contained, we introduce some distance measures between a state and the set $\adg$ and prove the continuity bounds for the state version as a comparison with $\Hat{E}_{\rm rev}^u(\rho_{AB})$.

\begin{definition}
Let $\rho_{AB}$ be a bipartite quantum state, the anti-degradable set distance is defined by
\begin{equation}
    d_{\rm set}(\rho_{AB}) := \min_{\sigma_{AB}\in {\rm ADG}} \frac{1}{2}\|\rho_{AB} - \sigma_{AB}\|_1,
\end{equation}
where the minimization ranges over all anti-degradable states on system AB. 
\end{definition}
Similarly, one also has the anti-degradable map distance as follows.
\begin{definition}[\!\!~\cite{Leditzky2017}]
Let $\rho_{AB}$ be a bipartite quantum state with purification $\phi_{ABE}$, the anti-degradable map distance is defined by
\begin{equation}
    d_{\rm map}(\rho_{AB}):= \min_{\mathcal{D}:E\rightarrow B} \frac{1}{2} \|\rho_{AB} - \mathcal{D}(\rho_{AE})\|_1\\,
\end{equation}
where $\rho_{AE} = \tr_B(\phi_{ABE})$ and the minimization ranges over all {\rm CPTP} maps $\mathcal{D}$.
\end{definition}

Both parameters can be computed via SDP, ranging from $[0,1]$, and are equal to $0$ iff $\rho_{AB}$ is anti-degradable. Similar to the idea in~\cite{Sutter2014} for channels and the proof techniques in~\cite{Leditzky2017}, we utilize the continuity bound of the conditional entropy in Lemma \ref{lem:cond_entr_continuity} proved by Winter~\cite{Winter_2016} to derive two continuity upper bounds on the one-way distillable entanglement, concerning the distance measures above. The proofs can be found in Appendix \ref{appendix:proof_conti_bound}. We denote $h(p)=-p\log p - (1-p)\log(1-p)$ as the binary entropy and $g(p):= (1+p)h(\frac{p}{1+p})$ as the bosonic entropy. 

\begin{proposition}\label{thm:anti_set_conti_bound}
For any bipartite state $\rho_{AB}$ with an anti-degradable set distance $\varepsilon_{\rm set}$, it satisfies
\begin{equation*}\label{Eq:set_bound}
    E_{D,\to}(\rho_{AB}) \leq E_{\rm SCB}(\rho_{AB}) := 2\varepsilon_{\rm set}\log(|A|) + g(\varepsilon_{\rm set}).
\end{equation*}
\end{proposition}
\begin{proposition}\label{thm:anti_map_conti_bound}
For any bipartite state $\rho_{AB}$ with an anti-degradable map distance $\varepsilon_{\rm map}$, it satisfies
\begin{equation*}\label{Eq:map_bound}
E_{D,\to}(\rho_{AB}) \leq E_{\rm MCB}(\rho_{AB}) := 4\varepsilon_{\rm map}\log(|B|) + 2g(\varepsilon_{\rm map}).
\end{equation*}
\end{proposition}

\subsection{Examples of noisy states}\label{sec:conti_bound_ex}
We now compare the performance of different continuity bounds and the Rains bound with $\Hat{E}_{\rm rev}^u(\cdot)$ by some concrete examples. Due to noise and other device imperfections, one usually obtains some less entangled states in practice rather than the maximally entangled ones. Such a disturbance can be characterized by some CPTP maps appearing in each party. Thus for the task of the distillation of the maximally entangled states under practical noises, we consider Alice and Bob are sharing pairs of maximally entangled states affected by bi-local noisy channels, i.e.,
\begin{equation}\label{Eq:bilocal}
    \rho_{A'B'} = \cN_{A\to A'}\otimes\cN_{B\to B'}(\Phi_{AB}).
\end{equation}
Moreover, we also provide the exploration of our bounds on the one-way distillable entanglement of random quantum states in Appendix~\ref{appendix:more_examples}.

\paragraph{Qubit system} Suppose Alice's qubit is affected by the qubit amplitude damping channel with Kraus operators
$K_0 = \ketbra{0}{0} + \sqrt{1-\gamma}\ketbra{1}{1}$, $K_1 = \sqrt{\gamma}\ketbra{0}{1}$, and Bob's qubit is affected by the depolarizing channel $\cD(\rho) = (1-p)\rho + p {I_2}/{2}$. Set $\gamma = 0.1$ and the noise parameter of depolarizing noise varies in the range $p\in[0.15,0.3]$. $E_{D,\to}(\rho_{A'B'})$ has upper bounds as functions of $p$ shown in Fig.~\ref{fig:qubit}.

\paragraph{Qutrit system} For the system with a dimension $d>2$, we consider the multilevel versions of the amplitude damping channel (MAD)~\cite{chessa2021quantum} as a local noise for Alice. The Kraus operators of a MAD channel in a $d$-dimensional system are defined by
\begin{equation}\label{Eq:MAD_kraus_1}
\begin{aligned}
    \hat{K}_{i j} &\equiv \sqrt{\gamma_{j i}}\ketbra{i}{j}, \quad \forall i, j \text { s.t. } 0 \leq i \leq j\leq d-1, \\
    \hat{K}_0 &\equiv \ketbra{0}{0}+ \sum_{1 \leq j \leq d-1} \sqrt{1-\xi_j}\ketbra{j}{j},
\end{aligned}
\end{equation}
with $\gamma_{j i}$ real quantities describing the decay rate from $j$-th to $i$-th level that fulfill the conditions
\begin{equation}\label{Eq:MAD_kraus_2}
\begin{cases}
    0 \leq \gamma_{j i} \leq 1, & \forall i, j \text { s.t. } 0 \leq i<j \leq d-1, \\ 
    \xi_j \equiv \sum_{0 \leq i<j} \gamma_{j i} \leq 1, & \forall j=1, \cdots, d-1.
\end{cases}
\end{equation}
Suppose Alice qutrit is affected by a MAD channel with $\gamma_{10} = \gamma_{20} = 0.1,  \gamma_{21} = 0$. Bob's qutrit is affected by a qutrit depolarizing channel with the noise parameter $p$. Then $E_{D,\to}(\rho_{A'B'})$ has upper bounds as functions of $p$ shown in Fig.~\ref{fig:qutrit}.

\paragraph{Qudit system} For the qudit system, we consider Alice's qudit is affected by a MAD channel with $d=4$ defined in Eq.~\eqref{Eq:MAD_kraus_1} and Eq.~\eqref{Eq:MAD_kraus_2}, where $\gamma_{10} = \gamma_{20} = \gamma_{30} = \gamma_{21} = 0.1, \gamma_{31} = \gamma_{32} = 0$.
Let Bob's qudit be affected by a qudit depolarizing channel with noise parameter $p$, then $E_{D,\to}(\rho_{A'B'})$ has upper bounds as functions of $p$ shown in Fig.~\ref{fig:qudit}.

\begin{figure*}[htbp]
\centering
    \begin{minipage}[t]{0.49\textwidth}
    \centering
    \includegraphics[width=0.9\linewidth]{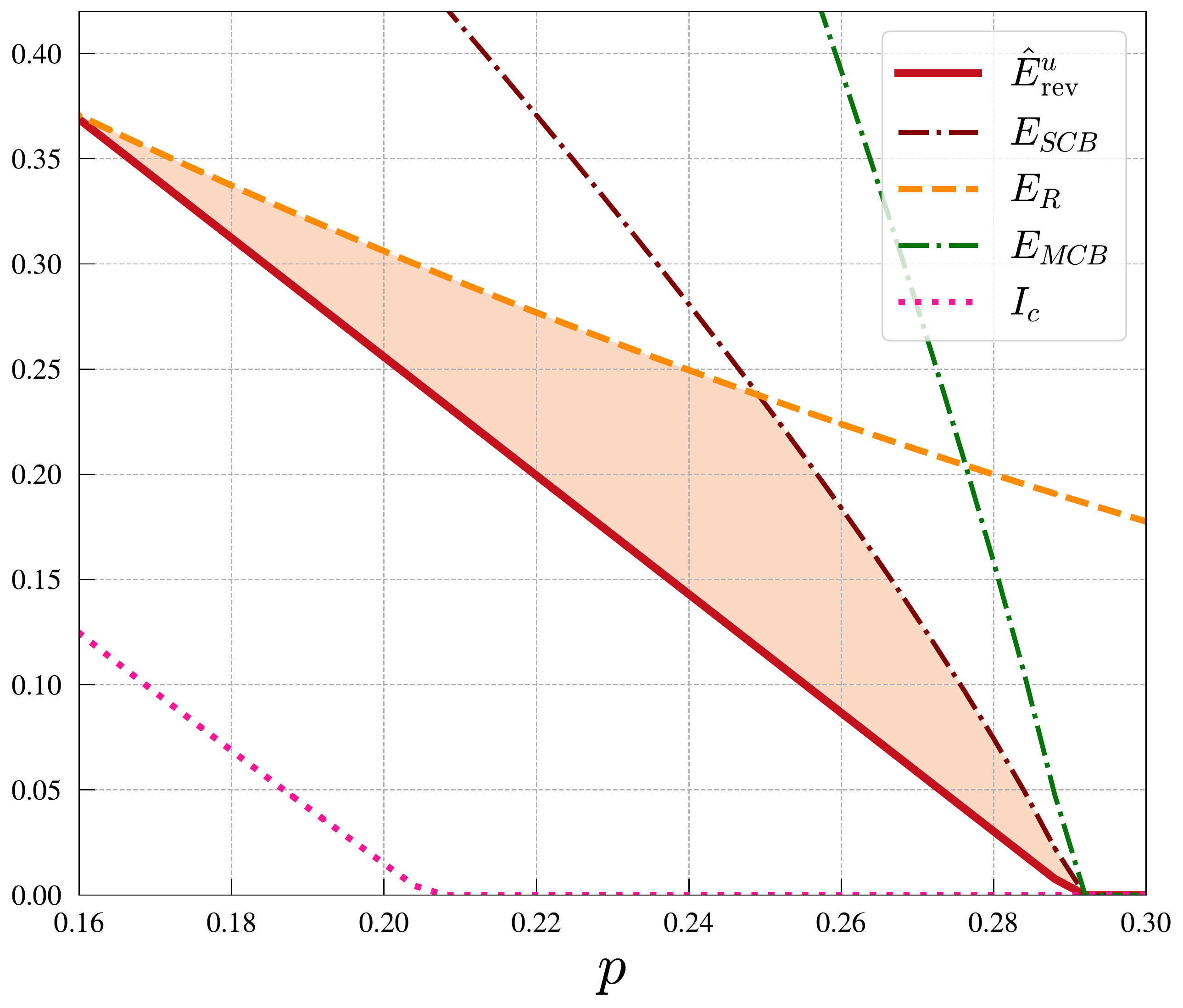}
    \caption{Upper bounds on the one-way distillable entanglement of the maximally entangled states affected by bi-local noises in a qubit system. The $x$-axis represents the change of the depolarizing noise $p$. The state's coherent information $I_c$ provides a lower bound. $R$ is the Rains bound. $\Hat{E}_{\rm rev}^u$ is the upper bound derived in Corollary \ref{coro:E_sqz_upperbound}. $E_{\rm SCB}$ and $E_{\rm MCB}$ are continuity bounds derived in Proposition \ref{thm:anti_set_conti_bound} and \ref{thm:anti_map_conti_bound}, respectively. It shows that $\Hat{E}_{\rm rev}^u$ outperforms all other upper bounds on these less entangled states.}
    \label{fig:qubit}
    \end{minipage}
    \hfill
    \begin{minipage}[t]{0.49\textwidth}
    \centering
    \includegraphics[width=0.9\linewidth]{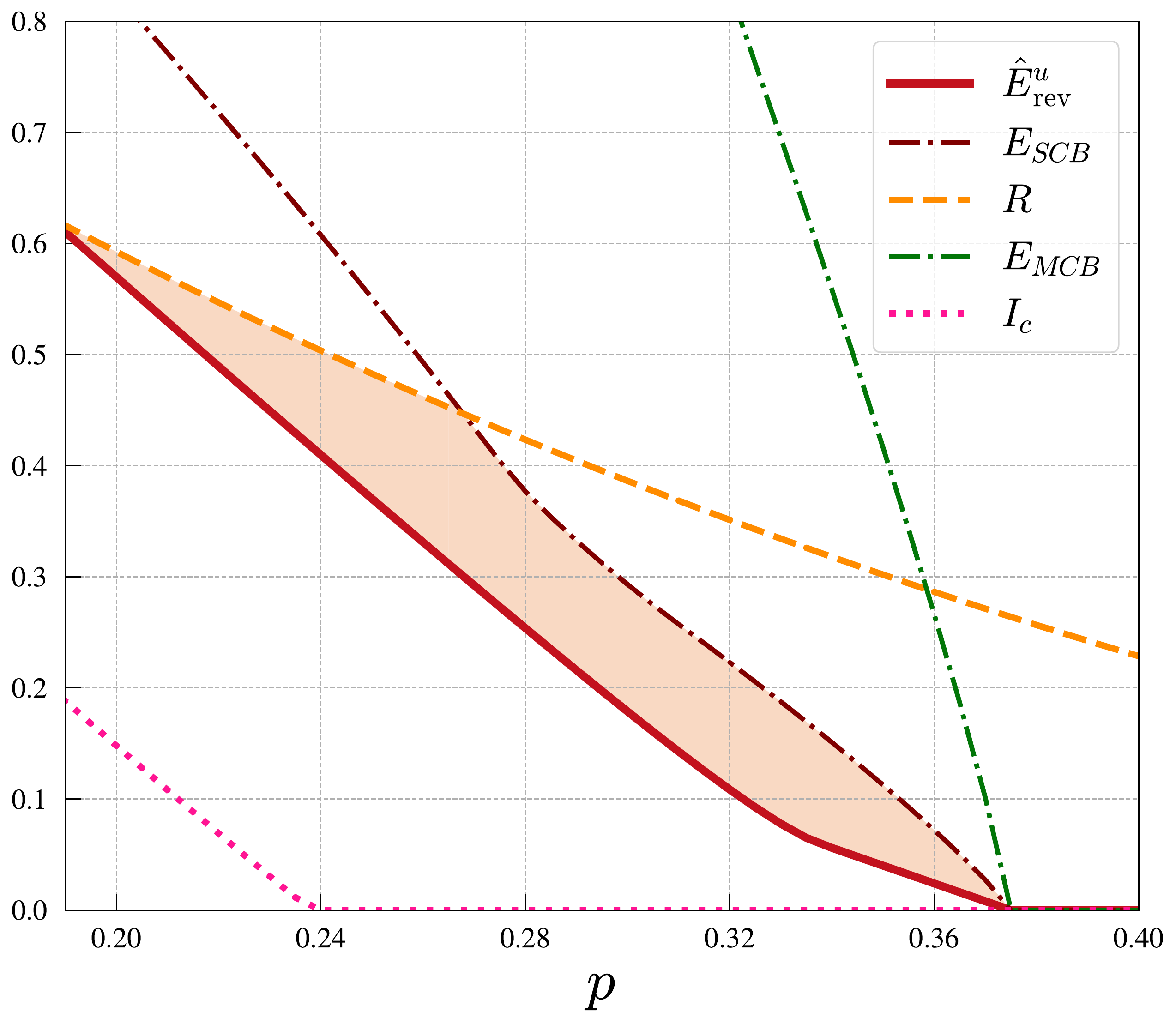}
    \caption{Upper bounds on the one-way distillable entanglement of the maximally entangled states affected by bi-local noises in a qutrit system. The $x$-axis represents the change of the depolarizing noise $p$. The parameters of the MAD channel is $\gamma_{10} = \gamma_{20} = 0.1,  \gamma_{21} = 0$.}
    \label{fig:qutrit}
    \end{minipage}
\end{figure*}
\begin{figure}
    \centering
    \includegraphics[width=0.9\linewidth]{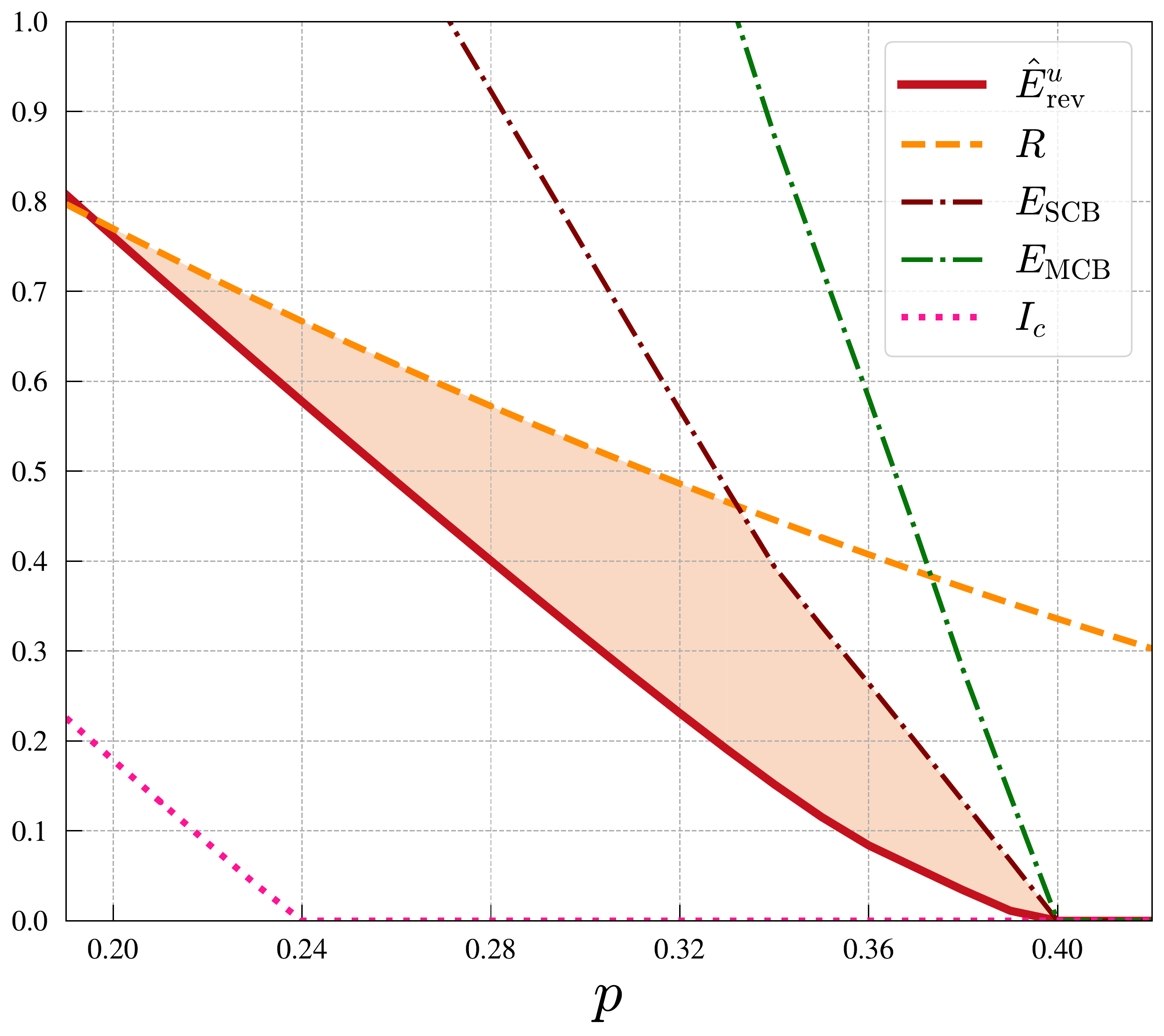}
    \caption{Upper bounds on the one-way distillable entanglement of the maximally entangled states affected by bi-local noises in a qudit system. The $x$-axis represents the change of the depolarizing noise $p$. The parameters of the MAD channel is $\gamma_{10} = \gamma_{20} = \gamma_{30} = \gamma_{21} = 0.1, \gamma_{31} = \gamma_{32} = 0$.}
    \label{fig:qudit}
\end{figure}

\subsection{Extending the method to the two-way distillable entanglement}\label{sec:two_way_app}
Similar to the \textit{reverse max-relative entropy of unextendible entanglement}, for a given bipartite state $\rho_{AB}$, we introduce the \textit{reverse max-relative entropy of NPT entanglement} as
\begin{equation}\label{Eq:sqz_mrppt_def}
    \cR_{\max,\PPT}(\rho_{AB}) := \min_{\tau\in \PPT}D_{\max}(\tau_{AB}||\rho_{AB}),
\end{equation}
where the minimization ranges over all $\PPT$ states. The term NPT refers to states whose partial transpose has negative eigenvalues. Based on the convexity of $E_{D,\leftrightarrow}(\cdot)$ on decomposing a state into maximally correlated (MC) states and PPT states~\cite{Leditzky2017} (see Appendix~\ref{appendix:two_way_DE_GBELL} for more details), we can utilize the reverse max-relative entropy of NPT entanglement to establish an upper bound on the two-way distillable entanglement, as outlined in Theorem~\ref{thm:twoway_bound}.
\begin{theorem}\label{thm:twoway_bound}
For any bipartite state $\rho_{AB}$, it satisfies
\begin{equation*}
E_{D,\leftrightarrow}(\rho_{AB}) \leq E_{\rm rev}^{npt}(\rho_{AB}):=[1-2^{-\cR_{\max,\PPT}(\rho_{AB})}]\cdot E_F(\omega_{AB}),
\end{equation*}
where $\omega_{AB}$ is the PPT-squeezed state of $\rho_{AB}$ and $\cR_{\max,\PPT}(\cdot)$ is the reverse max-relative entropy of NPT entanglement.
\end{theorem}
It also follows an efficiently computable relaxation as $\Hat{E}_{\rm rev}^{npt}(\rho_{AB})= [1-2^{-\cR_{\max,\PPT}(\rho_{AB})}] \sum_{i} \lambda_i S(B)_{\psi_i}$, where $\omega_{AB}=\sum_{i} \lambda_i \ketbra{\psi_i}{\psi_i}$ is the spectral decomposition of the PPT-squeezed state $\omega_{AB}$ of $\rho_{AB}$. 

In fact, $\Hat{E}_{\rm rev}^{npt}(\cdot)$ can be interpreted as an easily computable version of the bound $E_{\rm MP}(\cdot)$ in~\cite{Leditzky2017}, utilizing the convexity of $E_{D,\leftrightarrow}(\rho_{AB})$ on the convex decomposition of $\rho_{AB}$ into MC states and PPT states. Since the set of all MC states is not convex, tracking all possible decompositions to compute $E_{\rm MP}(\cdot)$ is generally hard. However, $\Hat{E}_{\rm rev}^{npt}(\cdot)$ is efficiently computable by SDP and provides a remarkably tighter approximation of $E_{\rm MP}(\cdot)$ for the example states presented in~\cite{Leditzky2017}. The comparison in detail can be found in Appendix \ref{appendix:two_way_DE_GBELL}. We note that $R(\rho_{AB}) \leq E_{\rm MP}(\rho_{AB})\leq \Hat{E}_{\rm rev}^{npt}(\rho_{AB})$ where $R(\cdot)$ is the Rains bound for the two-way distillable entanglement. Nevertheless, $\Hat{E}_{\rm rev}^{npt}(\cdot)$ connects the reverse max-relative entropy of NPT entanglement with the entanglement of formation, and we believe such connection would shed light on the study of other quantum resource theories as well.

\section{Applications on quantum channel capacity}\label{sec:app_on_QC}

For a general quantum channel $\cN_{A\rightarrow B}$, its quantum capacity has a regularized formula proved by Lloyd, Shor, and Devetak~\cite{Lloyd1997, Shor2002a, Devetak2005a}:
\begin{equation}\label{Eq:channel_cap}
    Q(\cN) = \lim_{n\rightarrow \infty}\frac{1}{n}Q^{(1)}(\cN^{\ox n}),
\end{equation}
where $Q^{(1)}(\cN):= \max_{\ket{\phi}_{A'A}} I_c(A'\rangle B)_{(\cI\ox \cN)(\phi)}$ is the channel coherent information. Similar to the one-way distillable entanglement of a state, the regularization in Eq.~\eqref{Eq:channel_cap} makes the quantum capacity of a channel intractable to compute generally. Substantial efforts have been made to establish upper bounds. One routine is inspired by the Rains bound in entanglement theory~\cite{Rains2001}. Tomamichel et al. introduced Rains information~\cite{Tomamichel2015a} for a quantum channel as an upper bound on the quantum capacity. Then some efficiently computable relaxations or estimations are given in~\cite{Wang2017d,Fang2019a}. Another routine is to consider the (anti)degradability of the quantum channels and to construct flag extensions, which gives the currently tightest upper bound for quantum channels with symmetry or known structures~\cite{Sutter2014, Wang2019b, Fanizza2019, Kianvash2022}. 

A channel $\cN_{A\rightarrow B}$ is called \textit{degradable} if there exists a CPTP map $\cD_{B\rightarrow E}$ such that $\cN^c = \cD\circ \cN$, and is called \textit{anti-degradable} if there exists a CPTP map $\cA_{E\rightarrow B}$ such that $\cN = \cA\circ \cN^c$. It is known that $\cN$ is (anti)degradable if and only if its Choi state $J_{\cN}$ is (anti)degradable. The quantum capacity of an anti-degradable channel is zero and the coherent information of a degradable channel is additive which leads to $Q(\cN) = Q^{(1)}(\cN)$. Concerning the (anti)degradability of a channel, the authors in \cite{Sutter2014} called a channel \textit{$\varepsilon$-degradable channel} if there is a CPTP map $\cD_{B\rightarrow E}$ such that $|| \cN^c - \cD\circ \cN||_{\diamond} \leq 2\varepsilon$. A channel is called \textit{$\varepsilon$-anti-degradable channel} if there is a CPTP map $\cA_{E\rightarrow B}$ such that $|| \cN - \cA\circ \cN^c||_{\diamond} \leq 2\varepsilon$. Based on these, one has continuity bounds of the quantum capacity as follows.
\begin{theorem}[\!\!\cite{Sutter2014}]\label{thm:conti_sutter}
    Given a quantum channel $\cN_{A \to B}$, if it is $\eps$-degradable, then it satisfies $Q(\cN) \leq Q^{(1)}(\cN) + \eps \log(d_E - 1) + h(\eps) + 2\eps\log d_E + g(\eps)$. If $\cN_{A\to B}$ is $\eps$-anti-degradable, it satisfies $Q(\cN) \leq \eps\log(|B|-1) + 2\eps\log|B| + h(\eps) + g(\eps)$.
\end{theorem}

With a similar spirit of the reverse max-relative entropy of unextendible entanglement in Eq.~\eqref{Eq:sqz_mradg_def}, we define the \textit{reverse max-relative entropy of anti-degradability of the channel} as
\begin{equation}\label{Eq:sqz_mrela_chan_def}
    \widetilde{\cR}_{\max,\adg}(\cN_{A\to B}) := \min_{\cN'_{A\to B}\in \cC_{\adg}}D_{\max}(\cN'_{A\to B}||\cN_{A\to B}),
\end{equation}
where $\cC_{\adg}$ is the set of all anti-degradable channels and the max-relative entropy of $\cN'_{A\to B}$ with respect to $\cN_{A\to B}$ is defined by
\begin{equation*}
    D_{\max}(\cN'_{A\to B}||\cN_{A\to B}):=  \inf \{\lambda \in \mathbb{R}: J_{AB}^{\cN'} \leq 2^{\lambda} J_{AB}^{\cN}\}.
\end{equation*}
If there is no such a channel $\cN'_{A\to B}\in \cC_{\adg}$ that satisfies $J_{AB}^{\cN'} \leq 2^{\lambda} J_{AB}^{\cN}$, $\widetilde{\cR}_{\max,\adg}(\cN_{A\to B})$ is set to be 0. Similar to the state case, $\widetilde{\cR}_{\max,\adg}(\cN_{A\to B})$ has a geometric implication analogous to the distance between $\cN_{A\to B}$ to the set of all anti-degradable channels. We can introduce the \textit{$\adg$-squeezed channel} of $\cN_{A\to B}$ as follows.
\begin{definition}\label{def:sqz_choi_state}
For a quantum channel $\cN_{A\rightarrow B}$ and the anti-degradable channel set $\cC_{\adg}$, if $\widetilde{\cR}_{\max,\adg}(\cN)$ is non-zero, the $\adg$-squeezed channel of $\cN_{A\rightarrow B}$ is defined by
\begin{equation}
    \cS_{A\to B} = \frac{\cN_{A\to B} - 2^{-\widetilde{\cR}_{\max,\adg}(\cN)} \cdot \cN^{\prime}_{A\to B}}{1-2^{-\widetilde{\cR}_{\max,\adg}(\cN)}}
\end{equation}
where $\cN^{\prime}_{A\to B}$ is the closest anti-degradable channel to $\cN_{A\to B}$ in terms of the max-relative entropy, i.e., the optimal solution in Eq.~\eqref{Eq:sqz_mrela_chan_def}. If $\widetilde{\cR}_{\max,\adg}(\cN)$ is zero, the $\adg$-squeezed channel of $\cN_{A\rightarrow B}$ is itself.
\end{definition}

Notably, $\widetilde{\cR}_{\max,\adg}(\cN)$ can be efficiently computed via SDP shown in Appendix~\ref{appendix:sqzbound_dual_sdp}. The conceptual idea we used here is similar to that for the state case in Eq.~\eqref{Eq:sqz_mrela_def} and Definition~\ref{def:sqz_state}, which is to squeeze or isolate out as much part of anti-degradable channel as possible via a convex decomposition of the original channel. The insight here is that one can ignore the contribution from the anti-degradable part for the quantum capacity, and the quantum capacity admits convexity on the decomposition into degradable and anti-degradable parts. In this way, the following Theorem~\ref{thm:general_chan_bound} gives an upper bound $Q_{\rm sqz}(\cN)$ for the quantum capacity of $\cN$.
\begin{theorem}\label{thm:general_chan_bound}
Given a quantum channel $\cN_{A \to B}$, if it has an ADG-squeezed channel $\cS_{A\to B}$, we denote
$\widehat{\cS}_{A\to BB'}$ as an extended channel of $\cS_{A\to B}$ such that $\tr_{B'}[\widehat{\cS}_{A\to BB'}(\rho_{A})] = \cS_{A\to B}(\rho_{A}), \forall \rho_A\in \cD(\cH_A)$. Then it satisfies
\begin{equation}
\begin{aligned}
    Q(\cN) \leq & Q_{\rm sqz}(\cN) := [1-2^{-\widetilde{\cR}_{\max,\adg}(\cN)}] \\
    &\cdot\min \left\{ Q^{(1)}(\widehat{\cS}) | \; \widehat{\cS}_{A\to BB'} \mbox{\rm \; is degradable}\right\},
\end{aligned}
\end{equation}
where the minimization is over all possible extended channels of $S_{A\to B}$. If there is no such a degradable $\widehat{\cS}_{A\to BB'}$ exists, the value of this bound is set to be infinity.
\end{theorem}
\begin{proof}
By the definition of the $\adg$-squeezed channel of $\cN_{A\to B}$, we have
\begin{equation*}
    \cN_{A\to B} = [1-2^{-\widetilde{\cR}_{\max,\adg}(\cN)}] \cdot \cS_{A\to B} + 2^{-\widetilde{\cR}_{\max,\adg}(\cN)}\cdot \cN'_{A\to B}
\end{equation*}
where $\cN'_{A\to B}$ is anti-degradable. We write an extended channel of $\cN'_{A\to B}$ as $\widehat{\cN'}_{A\to BB'}(\rho_A) = \cN_{A\to B}(\rho_A) \ox \ketbra{0}{0}_{B'}$, which is obviously anti-degradable. Then we can construct a quantum channel $\widehat{\cN}_{A\to BB'}$ as $
    \widehat{\cN}_{A\to BB'} = [1-2^{-\widetilde{\cR}_{\max,\adg}(\cN)}] \cdot \widehat{\cS}_{A\to BB'} + 2^{-\widetilde{\cR}_{\max,\adg}(\cN)}\cdot \widehat{\cN'}_{A\to BB'},$ such that $\tr_{B'}[\widehat{\cS}_{A\to BB'}(\rho_{A})] = \cS_{A\to B}(\rho_{A})$ for any state $\rho_A$ and $\widehat{\cS}_{A\to BB'}$ is degradable. This means after discarding the partial environment $B'$, the receiver can obtain the original quantum information sent through $\cN_{A\to B}$. In this case, $\widehat{\cN}$ can certainly convey more quantum information than the $\cN$, i.e., $Q(\cN)\leq Q(\widehat{\cN})$. Note that the quantum capacity admits convexity on the decomposition into degradable parts and anti-degradable parts~\cite{Wolf2007}. We conclude that
\begin{equation}
\begin{aligned}
    Q(\cN)\leq Q(\widehat{\cN}) &\leq [1-2^{-\widetilde{\cR}_{\max,\adg}(\cN)}] \cdot Q(\widehat{\cS})\\
    &\quad + 2^{-\widetilde{\cR}_{\max,\adg}(\cN)}\cdot Q(\widehat{\cN'})\\
    & = [1-2^{-\widetilde{\cR}_{\max,\adg}(\cN)}] \cdot Q^{(1)}(\widehat{\cS}),
\end{aligned}
\end{equation}
where the equality is followed by the quantum capacity is additive on degradable channels and is zero for anti-degradable channels. Considering the freedom of the choice of $\widehat{\cS}_{A\to BB'}$, we obtain
\begin{equation}
\begin{aligned}
    Q_{\rm sqz}(\cN) &:= [1-2^{-\widetilde{\cR}_{\max,\adg}(\cN)}] \\
    &\cdot \min \left\{ Q^{(1)}(\widehat{\cS}) |\; \widehat{\cS}_{A\to BB'} \mbox{\rm \; is degradable}\right\}
\end{aligned}
\end{equation}
as an upper bound on $Q(\cN)$.
\end{proof}
\begin{remark}
Theorem~\ref{thm:general_chan_bound} can be seen as a channel version of Theorem~\ref{thm:anti_sqz_bound}. However, in order to utilize the convexity of the quantum capacity after the squeezing process, it is challenging to decompose the ADG-squeezed channel into the sum of degradable ones. An alternative approach here is to use the idea of the extension channel. For the qubit channels specifically, this bound is efficiently computable and effective, as shown in subsection~\ref{sec:qc_pauli_chan}.
\end{remark}

\subsection{Quantum capacity of qubit channels}\label{sec:qc_pauli_chan}
For a quantum channel with dimension two in both input and output systems, we prove that the ADG-squeezed channel is always degradable. Thus, we give an efficiently computable upper bound on the quantum capacity using the idea of the reverse max-relative entropy of anti-degradability.
\begin{proposition}\label{prop:qubit_chan_sqzbound}
For any qubit channel $\cN_{A\to B}$, it is either anti-degradable or satisfies
\begin{equation}
\begin{aligned}
    Q(\cN) \leq& [1-2^{-\widetilde{\cR}_{\max,\adg}(\cN)}]\\
    &\cdot \max_{p \in[0,1]} I_{\mathrm{c}}\Big(p\ketbra{0}{0}+(1-p)\ketbra{1}{1}, \cS_{A\to B}\Big),
\end{aligned}
\end{equation}
where $I_{\mathrm{c}}(\rho, \mathcal{N}) \equiv H(\mathcal{N}(\rho))-H\left(\mathcal{N}^c(\rho)\right)$ and $\cS_{A\to B}$ is the ADG-squeezed channel of $\cN_{A\to B}$.
\end{proposition}
\begin{proof}
By the definition of the $\adg$-squeezed channel of $\cN_{A\to B}$, we have
\begin{equation*}
    \cN_{A\to B} = [1-2^{-\widetilde{\cR}_{\max,\adg}(\cN)}] \cdot \cS_{A\to B} + 2^{-\widetilde{\cR}_{\max,\adg}(\cN)}\cdot \cN'_{A\to B}
\end{equation*}
where $\cN'_{A\to B}$ is anti-degradable and $\cS_{A\to B}$ is not anti-degradable. If $\cS_{A\to B}$ is also not degradable, we can further decompose $\cS_{A\to B}$ into $\cS = c \cS_0 +(1-c)S_1$ such that $\cS_0$ is degradable and $\cS_1$ is anti-degradable since the extreme points of the set of all qubit channels have been shown to be degradable or anti-degradable channels~\cite{Cubitt_2008,Sutter2014}. This conflicts with the definition of $2^{-\widetilde{\cR}_{\max,\adg}(\cN)}$, which implies $\cS_{A\to B}$ is degradable. Thus,
\begin{equation*}
\begin{aligned}
    Q(\cN) &\leq [1-2^{-\widetilde{\cR}_{\max,\adg}(\cN)}] \cdot Q(\cS) + 2^{-\widetilde{\cR}_{\max,\adg}(\cN)} \cdot Q(\cN')\\
    &= [1-2^{-\widetilde{\cR}_{\max,\adg}(\cN)}] \cdot Q^{(1)}(\cS)\\
    &=[1-2^{-\widetilde{\cR}_{\max,\adg}(\cN)}]\\
    &\quad \cdot \max_{p \in[0,1]} I_{\mathrm{c}}\Big(p\ketbra{0}{0}+(1-p)\ketbra{1}{1}, \cS_{A\to B}\Big).
\end{aligned}
\end{equation*}
Note that the last equality is because $\cS_{A\to B}$ is degradable, and diagonal input states outperform non-diagonal states during the optimization of the channel coherent information~\cite{Wolf2007}.
\end{proof}

\begin{figure*}[t]
    \centering
    \includegraphics[width=\linewidth]{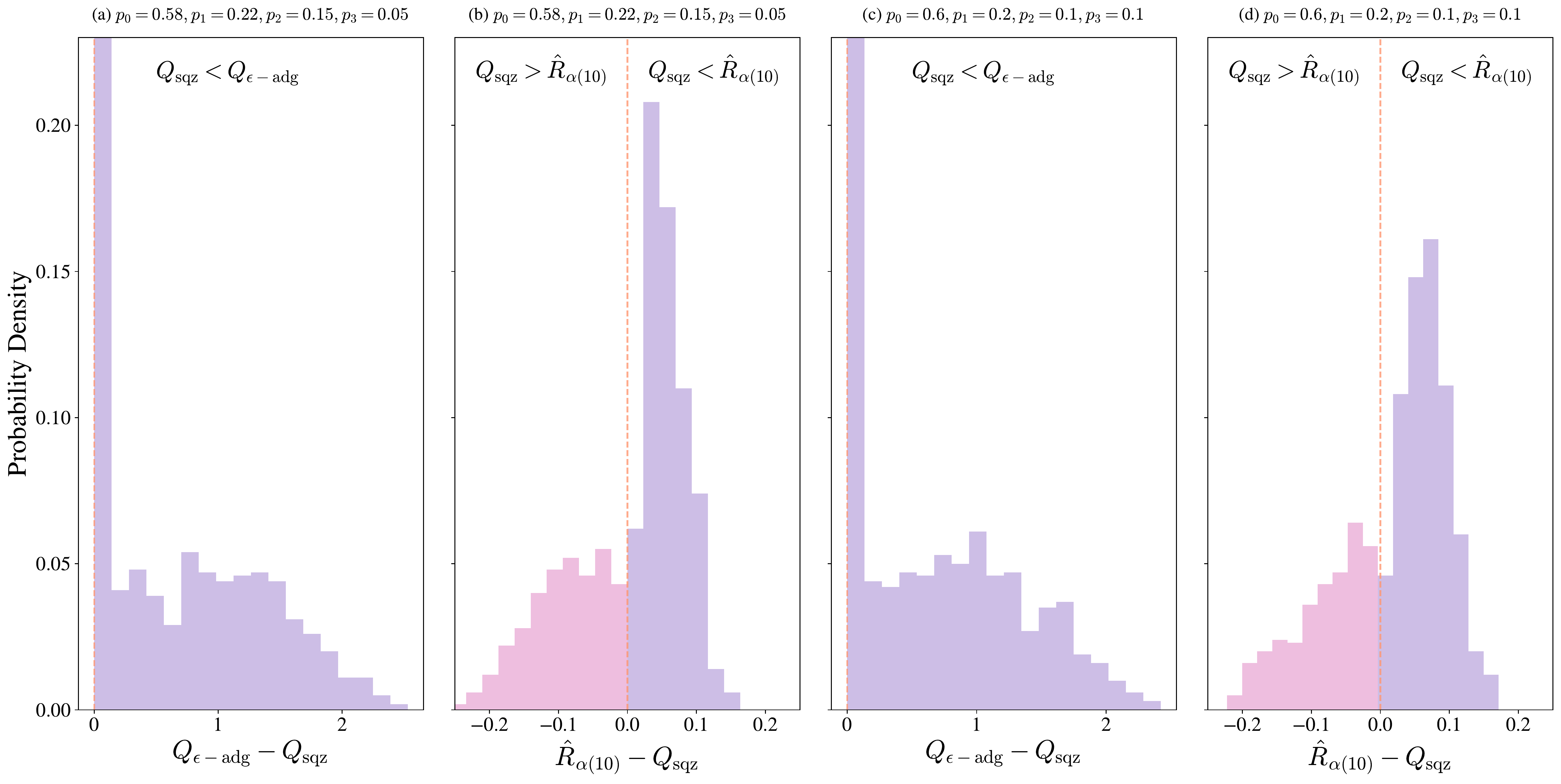}
    \caption{Upper bounds on the quantum capacity of the random mixed unitary channels. In Panel (a) and (b), we compare our bound $Q_{\rm sqz}$ in Proposition~\ref{prop:qubit_chan_sqzbound} with other upper bounds on the quantum capacity of 1000 randomly generated mixed unitary channels with fixed parameters $p_0 = 0.58$, $p_1 = 0.22$, $p_2 = 0.15$, $p_3 = 0.05$. The $x$-axis represents the distance between two bounds, and the $y$-axis represents the distribution of these channels. $Q_{\epsilon-\rm adg}$ is the continuity bound in Theorem~\ref{thm:conti_sutter} regarding the anti-degradability. $\hat{R}_{\alpha(10)}$ is the bound in~\cite{Fang2019a}. Panel (c) and (d) depict the same comparison for another 1000 randomly generated mixed unitary channels with fixed parameters $p_0 = 0.6$, $p_1 = 0.2$, $p_2 = 0.1$, $p_3 = 0.1$. In (a) and (c), we see that our bound always outperforms $Q_{\epsilon-\rm adg}$. In (b) and (d), we can see for many cases, our bound is tighter than $\hat{R}_{\alpha(10)}$.}
    \label{fig:mixU}
\end{figure*}
\paragraph{Mixed unitary channels} 
To compare the performance of our method with some best-known computable bounds, e.g., the continuity bound in Theorem~\ref{thm:conti_sutter} and the bound $\hat{R}_{\alpha}$~\cite{Fang2019a} generalized from the max-Rain information~\cite{Wang2017d}, we consider the mixed unitary channel $\cU_{A\rightarrow B}(\rho) = \sum_{i=0}^k p_i U_i \rho U_i^\dagger$, where $\sum_{i=0}^k p_i= 1$ and $U_i$ are unitary operators on a qubit system. In specific, we choose some fixed set of parameters and sample 1000 channels with randomly generated unitaries according to the Haar measure. We compute the distance between $Q_{\rm sqz}$ and other bounds, then obtain statistics on the distribution of these channels according to the distance value. The distribution results are depicted in Fig.~\ref{fig:mixU} where the purple region corresponds to the cases $Q_{\rm sqz}$ is tighter, and the pink region corresponds to the cases $Q_{\rm sqz}$ is looser. We can see that in Fig.~\ref{fig:mixU}(a) and Fig.~\ref{fig:mixU}(c), $Q_{\rm sqz}$ always outperforms the continuity bound of anti-degradability and in Fig.~\ref{fig:mixU}(b) and Fig.~\ref{fig:mixU}(d), our bound is tighter than $\hat{R}_{\alpha(10)}$ for many cases.

\begin{figure}[t]
    \centering
    \includegraphics[width=0.9\linewidth]{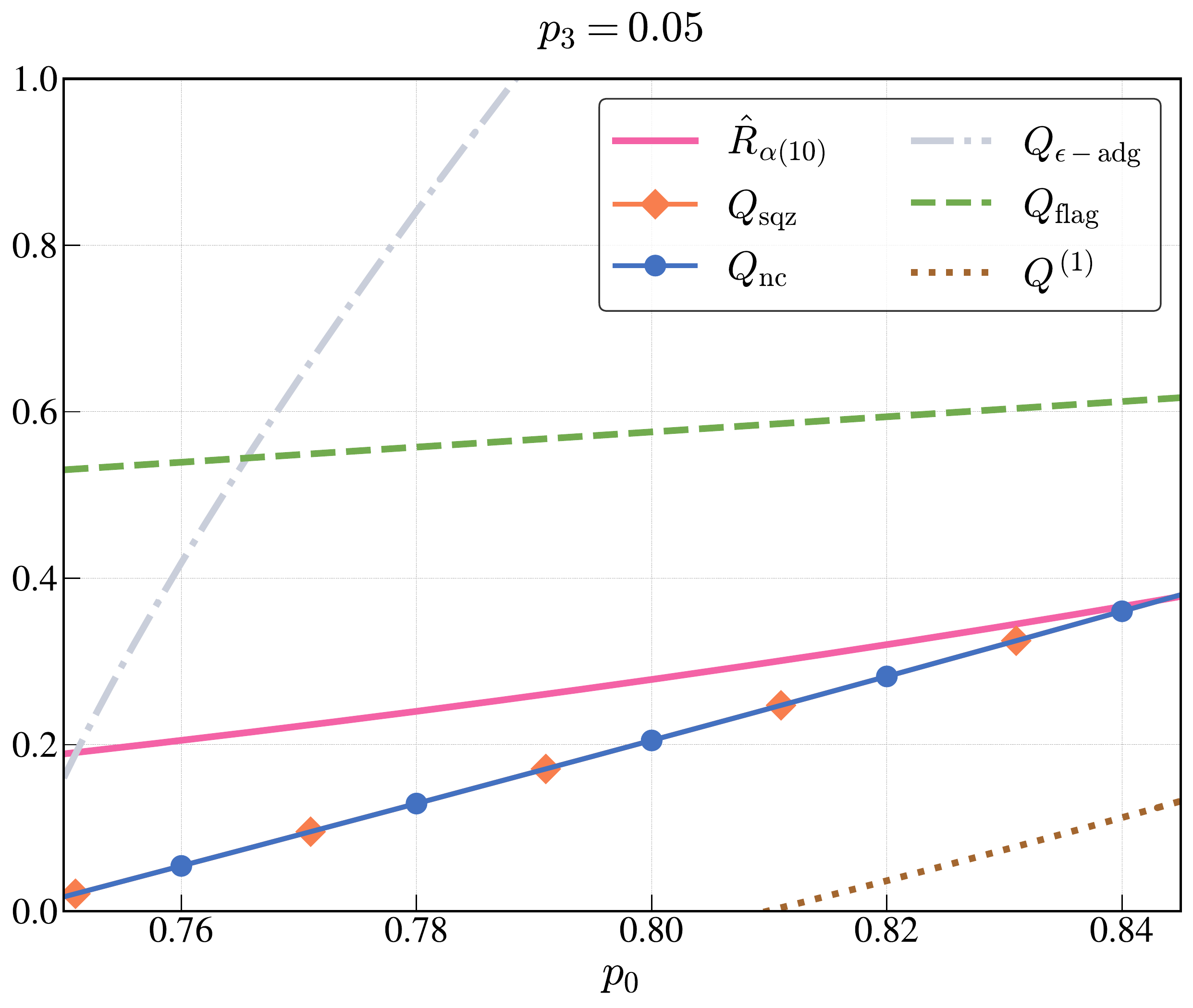}
    \caption{Upper bounds on quantum capacity of the covariant Pauli channels. The parameters are $p_3=0.05, p_0\in [0.75, 0.85]$. The single-shot coherent information $Q^{(1)}$~\cite{Bennett1996c,bausch2021error} provides a lower bound on the quantum capacity. $\hat{R}_{\alpha(10)}$ is the upper bound in~\cite{Fang2019a}. $Q_{\epsilon-\rm adg}$ is the continuity bound of anti-degradability. $Q_{\rm flag}$ is the upper bound by flag extension method in~\cite{poshtvan2022capacities}. $Q_{\rm sqz}$ is our bound in Corollary~\ref{cor:upper_bound_covPauli} and $Q_{\rm nc}$ is the no-cloning bound~\cite{Cerf2000}. It shows that $Q_{\rm sqz}$ can recover the no-cloning bound and outperform other computable bounds.}
    \label{fig:cov_Pauli}
\end{figure}

\paragraph{Pauli channels} A representative qubit channel is the Pauli channel describing bit-flip errors and phase-flip errors with certain probabilities in qubits. A qubit Pauli channel $\Lambda(\cdot)$ is defined as:
\begin{equation}\label{Eq:pauli_channel}
    \Lambda(\rho) = p_0 \rho + p_1 X\rho X + p_2 Y\rho Y +p_3 Z\rho Z,
\end{equation}
where $X,Y,Z$ are the Pauli operators and $\sum_{i=0}^3 p_i = 1$ are probability parameters. Note that for the quantum capacity, we only need to consider the cases where $p_0$ dominates. Since if, for example, bit flip error $X$ happens with probability larger than $p_0$, one can first apply a $X$ flip, mapping that channel back into the case where $p_0>p_1$. After utilizing our method on Pauli channels, the ADG-squeezed parameter is characterized in Proposition~\ref{prop:Pauli_chan_bound},  whose proof can be found in Appendix~\ref{appendix:pauli_chan_bound}. Thus, combined with Proposition~\ref{prop:qubit_chan_sqzbound}, we can recover the no-cloning bound~\cite{Cerf2000} on the quantum capacity of qubit Pauli channels.
\begin{proposition}\label{prop:Pauli_chan_bound}
For a qubit Pauli channel $\Lambda(\cdot)$ with $p_0\geq p_i>0$ $(i=1,2,3)$, it is either anti-degradable or satisfies
\begin{equation*}
    2^{-\widetilde{\cR}_{\max,\adg}(\Lambda)} = (\sqrt{p_1}+\sqrt{p_2})^2 + (\sqrt{p_2}+\sqrt{p_3})^2+(\sqrt{p_1}+\sqrt{p_3})^2
\end{equation*}
with an ADG-squeezed channel as the identity channel.
\end{proposition}

\begin{theorem}[\!\!~\cite{Cerf2000}]\label{thm:upper_bound_Pauli}
For a qubit Pauli channel $\Lambda(\cdot)$, its quantum capacity is either vanishing or satisfies 
\begin{equation*}
    Q(\Lambda) \leq 
        1-(\sqrt{p_1}+\sqrt{p_2})^2 -(\sqrt{p_2}+\sqrt{p_3})^2-(\sqrt{p_1}+\sqrt{p_3})^2.
\end{equation*}
\end{theorem}

\noindent One recent work \cite{poshtvan2022capacities} studies the capacities of a subclass of Pauli channels called the \textit{covariant Pauli channel}, where the parameters are set $p_1 = p_2$ with $p_0+2p_1+p_3=1$, i.e., $\Lambda_{\rm cov}(\rho) = p_0 \rho + p_1 (X\rho X + Y\rho Y)+ p_3 Z\rho Z$. Applying Theorem~\ref{thm:upper_bound_Pauli} on the covariant Pauli channels, we can bound their quantum capacity as follows.
\begin{corollary}\label{cor:upper_bound_covPauli}
For a covariant Pauli channel $\Lambda_{\rm cov}(\cdot)$, it is either anti-degradable with a zero quantum capacity or satisfies $Q(\Lambda_{\rm cov})\leq Q_{\rm sqz}(\Lambda_{\rm cov})$, where 
\begin{equation*}
    Q_{\rm sqz} (\Lambda_{\rm cov}) = 
         3p_0 + p_3 -\sqrt{8(p_3-p_0p_3-p_3^2)}-2.
\end{equation*}
\end{corollary}
In Fig.~\ref{fig:cov_Pauli}, we compare our bound with the upper bounds given in~\cite{poshtvan2022capacities} and the continuity bound of anti-degradability in Theorem~\ref{thm:conti_sutter}. It can be seen that our bound in the orange line, coinciding with the no-cloning bound, outperforms previous bounds, and thus can better characterize the quantum capacity of $\Lambda_{\rm cov}(\cdot)$ when it is close to being anti-degradable.

\section{Concluding remarks}
\label{sec:conclusion}

The design and implementation of a quantum internet~\cite{Kimble2008,Cacciapuoti2020,Wehner2018b,Illiano2022} involve a unique set of challenges and require advancements in the realms of entanglement generation, manipulation, and distribution, as well as the development of robust quantum communication protocols. By addressing the challenges associated with entanglement distillation and reliable quantum communication, our work contributes to the advancement of the field of quantum internet.

We have introduced the reverse max-relative entropy of entanglement which is related to the "weight of resource" in general resource theory. From a conceptual and technical side, the concept helps us to quantify how much useless entanglement we can squeeze out from a state or a channel, which is meaningful for the distillable entanglement and the quantum capacity, respectively. Importantly, these quantities can be efficiently computed via SDP, providing upper bounds on the distillable entanglement and quantum capacity that can be easily computed. 

To further investigate entanglement distillation, we have derived various continuity bounds on the one-way distillable entanglement using the anti-degradability of the state. In particular, our bound derived from the reverse max-relative entropy of unextendible entanglement outperforms both continuity bounds and the Rains bound in estimating the one-way distillable entanglement for maximally entangled states under certain relevant noises. These bounds are also the only known computable ones for general states. We further introduced the reverse max-relative entropy of NPT entanglement and established connections to prior results on the two-way distillable entanglement. For the quantum capacity of noisy channels, our bounds based on the reverse max-relative entropy of anti-degradability deliver improved results for random mixed unitary qubit channels. Also, the analytical bound obtained from our method recovers the no-cloning bound on Pauli channels~\cite{Cerf2000}.

These results open a novel way to connect valuable quantum resource measures with quantum communication tasks. Except for the existing applications of the reverse max-relative entropy of resources~\cite{Fang2020,Regula2021,Ducuara2020,Uola2020,Ducuara_2020,regula2022overcoming,regula2022postselected}, we expect that the reverse divergence of resources will find more applications in quantum resource theories in both asymptotic and non-asymptotic regimes~\cite{Tomamichel2015book,Fang2017,Datta2009,Wang2012,Regula2019,WW19}. Our method may also be generalized to quantum entanglement and quantum communication theories in continuous-variable quantum information~\cite{Braunstein2005a}. These advancements hold promise to facilitate the development of practical, scalable, and efficient techniques for generating, preserving, and utilizing entangled states over long distances, thereby advancing the field of quantum internet.

\section{Acknowledgements.} 
\noindent C.Z and C.Z contributed equally to this work.  We would like to thank Bartosz Regula and Ludovico Lami for their helpful comments. 
This work was partially supported by the Start-up Fund (No. G0101000151) from The Hong Kong University of Science and Technology (Guangzhou), the Guangdong Provincial Quantum Science Strategic Initiative (No. GDZX2303007), the Innovation Program for Quantum Science and Technology  (No. 2021ZD0302901), the Quantum Science Center of Guangdong–Hong Kong–Macao Greater Bay Area, and the Education Bureau of Guangzhou Municipality.

\normalsize
\bibliography{main}

\appendices

\vspace{2cm}

\setcounter{subsection}{0}
\setcounter{table}{0}
\setcounter{figure}{0}

\renewcommand{\theequation}{S\arabic{equation}}
\numberwithin{equation}{section}
\renewcommand{\theproposition}{S\arabic{proposition}}
\renewcommand{\thedefinition}{S\arabic{definition}}
\renewcommand{\thefigure}{S\arabic{figure}}
\setcounter{equation}{0}
\setcounter{table}{0}
\setcounter{section}{0}
\setcounter{proposition}{0}
\setcounter{definition}{0}
\setcounter{figure}{0}

\section{Dual SDP for $\cR_{\max,\adg}(\rho_{AB})$ and $\widetilde{\cR}_{\max,\adg}(\cN)$}\label{appendix:sqzbound_dual_sdp}
The primal SDP for calculating $1-2^{-\cR_{\max,\adg}(\rho_{AB})}$ of the state $\rho_{AB}$ can be written as:
\begin{subequations}
\begin{align}\label{sdp:primal_state}
\min_{\omega_{AB}, \tau_{AB}, \tau_{ABE}}&  \tr[\omega_{AB}],\\
 {\rm s.t.} &\; \rho_{AB} =\omega_{AB}+\tau_{AB}, \\
&\; \omega_{AB} \geq 0, \tau_{AB} \geq 0, \tau_{ABE} \geq 0,\\
&\; \tr_{E}[\tau_{ABE}] = \tr_{B}[\tau_{ABE}] = \tau_{AB}, \label{Eq:extendible}
\end{align}
\end{subequations}
where Eq.~\eqref{Eq:extendible} corresponds to the anti-degradable condition of $\tau_{AB}$. The Lagrange function of the primal problem is
\begin{equation*}
\begin{aligned}
    &L(\omega_{AB},\tau_{AB},\tau_{ABE},M, N, K)\\ 
    &= \tr[\omega_{AB}] + \langle M, \rho_{AB}-\omega_{AB}-\tau_{AB} \rangle \\
    &\quad + \langle N, \tr_{E}[\tau_{ABE}]-\tau_{AB}\rangle + \langle K, \tr_{B}[\tau_{ABE}]-\tau_{AB}\rangle\\
    &= \langle M, \rho_{AB}\rangle + \langle I-M, \omega_{AB}\rangle + \langle -M-N-K, \tau_{AB}\rangle \\
    &\quad + \langle N\otimes I_E + P_{BE}(K_{AB} \otimes I_E)P_{BE}^\dagger, \tau_{ABE}\rangle,
\end{aligned}
\end{equation*}
where $M_{AB}, N_{AB}, K_{AB}$ are Lagrange multipliers and $P_{BE}$ is the permutation operator between $B$ and $E$. The corresponding Lagrange dual function is 
\begin{equation*}
    g(M,N,K) = \inf_{\omega_{AB},\tau_{AB}, \tau_{ABE}\geq 0} L(\omega_{AB},\tau_{AB},\tau_{ABE},M, N, K).
\end{equation*}
Since $\omega_{AB} \geq 0, \tau_{AB}\geq 0, \tau_{ABE}\geq 0$, it must hold that $I-M_{AB}\geq 0, -M-N-K\geq 0, N\otimes I + P_{BE}(K \otimes I)P_{BE}^\dagger \geq 0$. Thus the dual SDP is
\begin{equation*}
\begin{aligned}\label{sdp:state_dual}
\max_{M_{AB},N_{AB},K_{AE}}&  \tr[M_{AB}\rho_{AB}],\\
 {\rm s.t.} &\; M_{AB} \leq I_{AB}, \\
&\; M_{AB}+N_{AB}+K_{AB} \leq 0,\\
&\; N_{AB}\otimes I_{E} +P_{BE}(K_{AB} \otimes I_{E})P_{BE}^\dagger \geq 0.
\end{aligned}
\end{equation*}

The primal SDP for calculating $1-2^{-\widetilde{\cR}_{\max,\adg}(\cN)}$ of the channel $\cN_{A\to B}$ is:
\begin{subequations}\label{sdp:chan_primal}
\begin{align}
\min_{\Gamma_{AB}^\cS, \Gamma_{AB}^{\cN'}, \gamma_{ABE}}& \tr[\Gamma_{AB}^\cS],\\
 {\rm s.t.} &\; J_{AB}^\cN = \Gamma_{AB}^\cS + \Gamma_{AB}^{\cN'}, \\
&\; \Gamma_{AB}^\cS \geq 0, \Gamma_{AB}^{\cN'} \geq 0, \gamma_{ABE} \geq 0,\\
&\; \tr_{B}[\Gamma_{AB}^{\cN'}] = \tr[\Gamma_{AB}^{\cN'}]/d_A \cdot I_{A},\\
&\; \tr_{E}[\gamma_{ABE}] = \tr_{B}[\gamma_{ABE}] = \Gamma_{AB}^{\cN'}, \label{Eq:choi_extendible}
\end{align}
\end{subequations}
where Eq.~\eqref{Eq:choi_extendible} corresponds to the anti-degradable condition of the unnormalized Choi state $\Gamma_{AB}^{\cN'}$. The Lagrange function of the primal problem is
\begin{equation*}
\begin{aligned}
    &L(\Gamma_{AB}^\cS, \Gamma_{AB}^{\cN'}, \gamma_{ABE},M, N, K, R)\\ 
    &= \tr[\Gamma_{AB}^\cS] + \langle M, J_{AB}^\cN-\Gamma_{AB}^\cS-\Gamma_{AB}^{\cN'} \rangle \\
    &\quad + \langle N, \tr_{E}[\gamma_{ABE}]-\Gamma_{AB}^{\cN'}\rangle + \langle K, \tr_{B}[\gamma_{ABE}]-\Gamma_{AB}^{\cN'}\rangle\\
    &\quad + \langle R, \tr_{B}[\Gamma_{AB}^\cS]-\tr[\Gamma_{AB}^\cS]/d_A \cdot I_{A}\rangle \\
    &= \langle M, J_{AB}^\cN\rangle + \langle (1-\tr R/d_A)I-M+R\otimes I_B, \Gamma_{AB}^\cS\rangle \\
    &\quad + \langle -M-N-K, \Gamma_{AB}^{\cN'}\rangle\\
    & \quad + \langle N\otimes I_E + P_{BE}(K_{AE} \otimes I_B)P_{BE}^\dagger, \gamma_{ABE}\rangle,
\end{aligned}
\end{equation*}
where $M_{AB}, N_{AB}, K_{AB}$ are Lagrange multipliers and $P_{BE}$ is the swap operator between the system $B$ and $E$. The corresponding Lagrange dual function is 
\begin{equation*}
    g(M,N,K) = \inf_{\Gamma_{AB}^\cS,\Gamma_{AB}^{\cN'}, \gamma_{ABE}\geq 0} L(\Gamma_{AB}^\cS,\Gamma_{AB}^{\cN'},\gamma_{ABE},M, N, K).
\end{equation*}
Since $\Gamma_{AB}^\cS \geq 0, \Gamma_{AB}^{\cN'}\geq 0, \gamma_{ABE}\geq 0$, it must hold that 
\begin{equation}
\begin{aligned}
    (1-\tr R_A/d_A)I_{AB}-M_{AB}+R_A\otimes I_B &\geq 0\\
    -M_{AB}-N_{AB}-K_{AB} &\geq 0\\
    N_{AB}\otimes I_E + P_{BE}(K_{AB} \otimes I_{E})P_{BE}^\dagger &\geq 0.
\end{aligned}
\end{equation}
Thus the dual SDP is
\begin{equation}
\label{equ:dual_sdp_channel}
\begin{aligned}
\max_{M_{AB},N_{AB},K_{AB},R_{A}}&  \tr[M_{AB} J_{AB}^\cN],\\
 {\rm s.t.} &\; M_{AB} \leq (1-\tr R_A/d_A)I_{AB}+R_A\otimes I_B, \\
&\; M_{AB}+N_{AB}+K_{AB} \leq 0,\\
&\; N_{AB}\otimes I_{E} +P_{BE}(K_{AB} \otimes I_E)P_{BE}^\dagger \geq 0.
\end{aligned}
\end{equation}

\section{Two-way distillable entanglement}\label{appendix:two_way_DE_GBELL}
We first start with the definition of the maximally correlated (MC) state. 
\begin{definition}
A bipartite state $\rho_{A B}$ on $\mathbb{C}^d \times \mathbb{C}^d$ is said to be maximally correlated (MC), if there exist bases $\left\{|i\rangle_A\right\}_{i=0}^{d-1}$ and $\left\{|i\rangle_B\right\}_{i=0}^{d-1}$ such that $\rho_{A B}=\sum_{i, j=0}^{d-1} \alpha_{i j}\ketbra{i}{j}_A \otimes \ketbra{i}{j}_B$,
where $\left(\alpha_{i j}\right)$ is a positive semidefinite matrix with trace 1.
\end{definition}
By Lemma~\ref{lem:MC_PPT}, we arrive at the upper bound in Theorem \ref{thm:twoway_bound}, as every pure state is an MC state.
\begin{lemma}[\!\!~\cite{Leditzky2017}]\label{lem:MC_PPT}
    The two-way distillable entanglement is convex on convex combinations of {\rm MC} and {\rm PPT} states.
\end{lemma}

We then consider the two-way distillable entanglement of the example states $\rho_{AB}$ given in Eq.(3.7) in~\cite{Leditzky2017}. We plot different upper bounds on $E_{D,\leftrightarrow}(\rho_{AB})$ in Fig.~\ref{fig:MC_ppt}. Note that the result presented in~\cite{Leditzky2017} is actually an approximation of $E_{\rm MP}(\cdot)$, as it does not consider the minimization over all possible decompositions. From the plot, we observe that our bound $\Hat{E}_{\rm rev}^{npt}(\rho_{AB})$ is tighter than both $E_{\rm W}(\cdot)$ and the approximation of $E_{\rm MP}(\cdot)$ as a function of $p$.

\begin{figure}[t]
    \centering
    \includegraphics[width=0.9\linewidth]{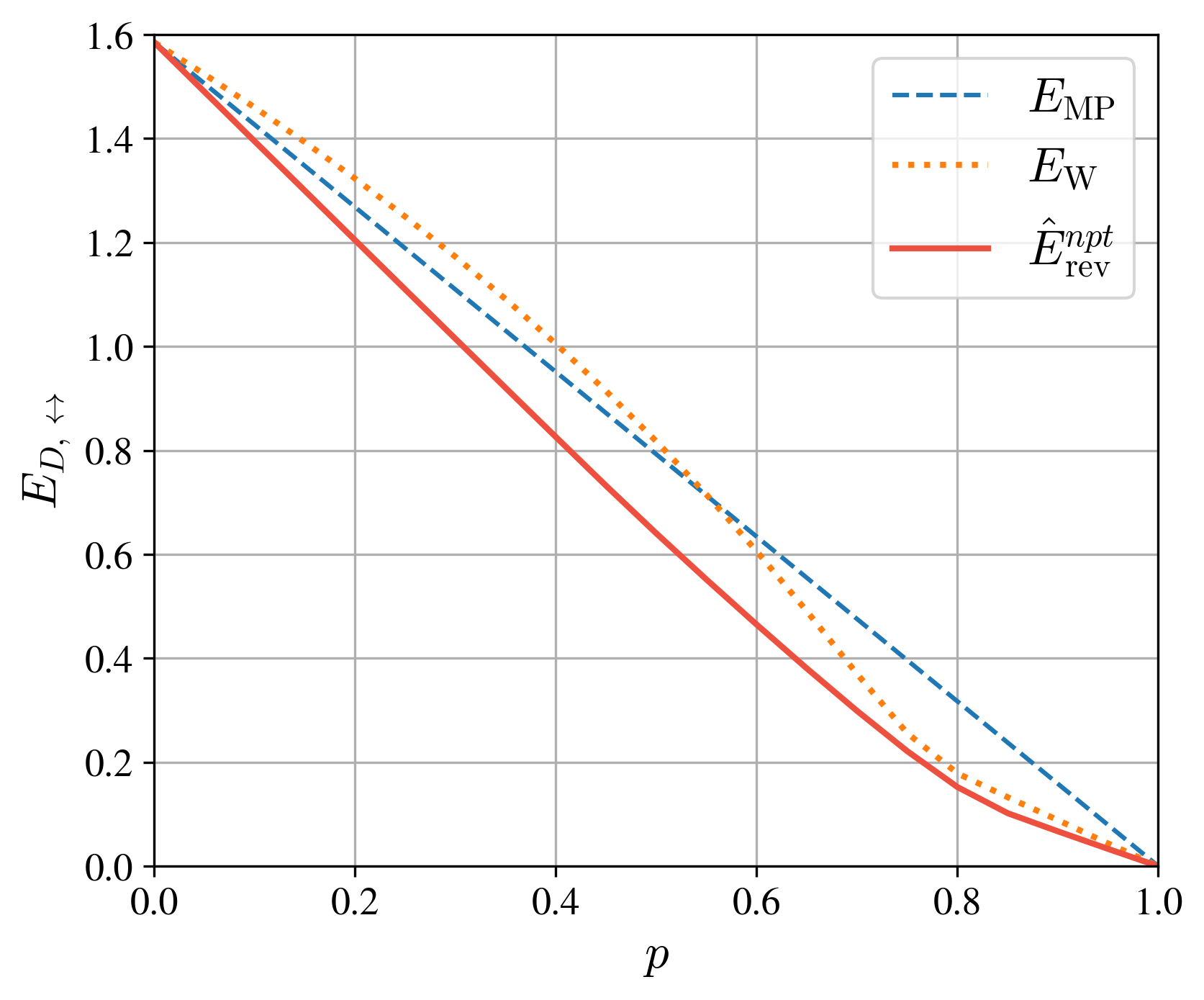}
    \caption{Upper bounds on the two-way distillable entanglement. $\hat{E}_\mathrm{rev}^{npt}$ is our bound derived from the reverse max-relative entropy of NPT entanglement. $E_{\rm W}$ is the SDP bound in~\cite{Wang2016}. $E_{\rm MP}$ is the upper bound in~\cite{Leditzky2017} based on the direct division of $\rho_{AB}$ into the MC states and the PPT states. }
    \label{fig:MC_ppt}
\end{figure}

\section{Proof of Proposition~\ref{thm:anti_set_conti_bound} and Proposition~\ref{thm:anti_map_conti_bound}}\label{appendix:proof_conti_bound}
\begin{lemma}[\!\!\cite{Winter_2016}]\label{lem:cond_entr_continuity}
For any states $\rho_{AB}$ and $\sigma_{AB}$ such that $\frac{1}{2}\|\rho_{AB}-\sigma_{AB}\|_1 \leq \varepsilon \leq 1$, it satisfies
\begin{equation}
    |S(A|B)_{\rho} - S(A|B)_{\sigma}| \leq 2\varepsilon\log(|A|) + (1+\varepsilon)h(\frac{\varepsilon}{1+\varepsilon}).
\end{equation}
\end{lemma}

\renewcommand\theproposition{\ref{thm:anti_set_conti_bound}}
\setcounter{proposition}{\arabic{proposition}-1}
\begin{proposition}
For any bipartite state $\rho_{AB}$ with an anti-degradable set distance $\varepsilon_{\rm set}$, it satisfies
\begin{equation}
    E_{D,\to}(\rho_{AB}) \leq 2\varepsilon_{\rm set}\log(|A|) + (1+\varepsilon_{\rm set})h(\frac{\varepsilon_{\rm set}}{1+\varepsilon_{\rm set}}).
\end{equation}
\end{proposition}
\begin{proof}
Since $\rho_{AB}$ has a anti-degradable set distance $\varepsilon_{\rm set}$, we denote $\sigma_{AB}$ the anti-degradable state with $\frac{1}{2}\|\rho_{AB} - \sigma_{AB}\|_1 \leq \varepsilon_{\rm set}$. Let $T:A^n \rightarrow A'M$ be an instrument with isometry $U_n:A^n \rightarrow A'MN$ and denote $\Delta = 2\varepsilon_{\rm set}\log(|A|) + (1+\varepsilon_{\rm set})h(\frac{\varepsilon_{\rm set}}{1+\varepsilon_{\rm set}})$. Then we have
\begin{subequations}\label{Eq:E1_rho^n}
\begin{align}
    E^{(1)}_{D,\to}(\rho_{AB}^{\otimes n}) &= \max_{U_n} I_c(A'\rangle B^n M)_{U_n \rho^{\otimes n} U_n^\dagger}\\
    & = \max_{U_n} - S(A' | B^n M)_{U_n \rho^{\otimes n} U_n^\dagger}\\
    &\leq \max_{U_n} - S(A' | B^n M)_{U_n (\sigma\ox \rho^{\otimes n-1}) U_n^\dagger} + \Delta \label{Eq:ineq_one_copy}\\
    & \leq \max_{U_n} - S(A' | B^n M)_{U_n \sigma^{\otimes n} U_n^\dagger} + n\Delta \label{Eq:ineq_n_copy}\\
    & = n\Delta,
\end{align}
\end{subequations}
where Eq.~\eqref{Eq:ineq_one_copy} follows by the fact that $\|\rho_{0}\otimes \rho_{1} - \sigma_{0}\otimes \sigma_{1}\|_1 \leq  \|\rho_{0} - \sigma_{0}\|_1 + \|\rho_{1} - \sigma_{1}\|_1$ and Lemma \ref{lem:cond_entr_continuity}. The inequality in Eq.~\eqref{Eq:ineq_n_copy} follows by applying the same argument $n$ times considering $\sigma^{\ox i}\rho^{\ox n-i}$ for $i=1,...,n$. After dividing Eq.~\eqref{Eq:E1_rho^n} by $n$ and taking the limit $n\rightarrow \infty$, we arrive at $E_{D,\to}(\rho_{AB}) = \lim_{n\rightarrow \infty} \frac{1}{n}E^{(1)}_{D,\to}(\rho_{AB}^{\otimes n}) \leq \Delta$.
\end{proof}
\renewcommand{\theproposition}{\arabic{proposition}}

\renewcommand\theproposition{\ref{thm:anti_map_conti_bound}}
\setcounter{proposition}{\arabic{proposition}-1}
\begin{proposition}
For any bipartite state $\rho_{AB}$ with an anti-degradable map distance $\varepsilon_{\rm map}$, it satisfies
\begin{equation*}
    E_{D,\to}(\rho_{AB}) \leq 4\varepsilon_{\rm map}\log(|B|) + 2(1+\varepsilon_{\rm map})h(\frac{\varepsilon_{\rm map}}{1+\varepsilon_{\rm map}}).
\end{equation*}
\end{proposition}

\begin{proof}
Let $\phi_{ABE}$ be a purification of $\rho_{AB}$, $\cD:E\rightarrow B$ be the CPTP map such that $\frac{1}{2}\|\rho_{AB} - \cD(\rho_{AE})\|_1 \leq \varepsilon_{\rm map}$ with an isometry $W:E\rightarrow B'G$. Let $T:A^n \rightarrow A'M$ be an instrument with isometry $U_n:A^n \rightarrow A'MN$ and denote $\Delta = 2\varepsilon_{\rm map}\log(|B|) + (1+\varepsilon_{\rm map})h(\frac{\varepsilon_{\rm map}}{1+\varepsilon_{\rm map}})$. For $t=1,2,...,n$, we can define pure states
\begin{equation*}
\begin{aligned}
\psi_{A^n B^n B'_1 G_1...B'_t G_t E_{t+1}...E_n}^t &= (W_1\otimes \cdots \otimes W_t) \phi_{ABE}^n\\
\theta_{A' M N B^n B'_1 G_1...B'_t G_t E_{t+1}...E_n}^t &= U_n \psi^t,\\
\omega_{A' M N B^n E^n} &= U_n \phi_{ABE}^n
\end{aligned}
\end{equation*}
We further define $\hat{\rho}_{AB'} = \cD(\rho_{AE})$ which shares the same purification with $\rho_{AB}$ listed above, thus an anti-degradable state. Then for $t=n$ we have $\theta^n = U_n (W_1\otimes \cdots \otimes W_n) U_n^\dagger \omega$, it yields 
\begin{equation*}\label{Eq:I_c}
\begin{aligned}
    I_c(A'\rangle B^n M)_{\omega} &=  I_c(A'\rangle B^n M)_{\theta^n_{A' M N B^n B'^n G^n}}\\
    &= S(B^n M)_{\theta} - S(A'B^n M)_{\theta} \\
    &= S(A' N B'^n G^n)_{\theta} - S(A'B^n M)_{\theta}\\
    &= S(A' N B'^n G^n)_{\theta} - S(A'B'^n M)_{\theta} \\
    &\quad + S(A' B'^n M)_{\theta} - S(A' B^n M)_{\theta}\\
    &= S(G^n|A'B'^n M)_{\theta} + S(A' B'^n M)_{\theta} - S(A' B^n M)_{\theta}
\end{aligned}
\end{equation*}
where we abbreviate $\theta = \theta^n_{A' M N B^n B'^n G^n}$. Applying the same technique in the proof of Theorem 2.12 in~\cite{Leditzky2017}, we can bound $S(A' B'^n M)_{\theta} - S(A' B^n M)_{\theta} \leq n\Delta$.
Consequently, it follows that
\begin{equation}\label{Eq:E2_rho^n}
\begin{aligned}
    I_c(A'\rangle B^n M)_{\omega} &\leq S(G^n|A'B'^n M)_{\theta}  + n\Delta\\
    &\leq S(G^n|B'^n)_{\theta}  + n\Delta \\
    &= S(G^n B'^n)_{\theta} - S(B'^n)_{\theta}  + n\Delta\\
    &\leq S(G^n B'^n)_{\theta} - S(A'M N G^n B'^n)_{\theta}  + 2n\Delta\\
    &\leq I(A'M N \rangle G^n B'^n)_{\theta}  + 2n\Delta\\
    &= I(A^n \rangle B'^n)_{\hat{\rho}_{AB'}^{\ox n}}  + 2n\Delta \\
    &= n [I(A \rangle B)_{\hat{\rho}_{AB'}}+ 2\Delta] = 2n\Delta,
\end{aligned}
\end{equation}
where the last equality is due to the anti-degradability of $\hat{\rho}_{AB'}$. After dividing Eq.~\eqref{Eq:E2_rho^n} by $n$ and taking the limit $n\rightarrow \infty$, we arrive at $E_{D,\to}(\rho_{AB}) = \lim_{n\rightarrow \infty} \frac{1}{n}E^{(1)}_{D,\to}(\rho_{AB}^{\otimes n}) \leq 2\Delta$.
\end{proof}
\renewcommand{\theproposition}{S\arabic{proposition}}

\section{Proof of Proposition \ref{prop:Pauli_chan_bound}}\label{appendix:pauli_chan_bound}

\begin{lemma}[\!\!\cite{chen2014symmetric}]\label{lem:check_adg}
A two qubit state $\rho_{A B}$ is anti-degradable if and only if,
\begin{equation}\label{Eq:checkadg}
\operatorname{Tr}\left(\rho_B^2\right) \geq \operatorname{Tr}\left(\rho_{A B}^2\right)-4 \sqrt{\operatorname{det}\left(\rho_{A B}\right)}.
\end{equation}
\end{lemma}

\renewcommand\theproposition{\ref{prop:Pauli_chan_bound}}
\setcounter{proposition}{\arabic{proposition}-1}

\begin{proposition}
For a qubit Pauli channel $\Lambda(\cdot)$ with $p_0\geq p_i>0 (i=1,2,3)$, it is either anti-degradable or satisfies
\begin{equation}
    2^{-\widetilde{\cR}_{\max,\adg}(\Lambda)} = (\sqrt{p_1}+\sqrt{p_2})^2 + (\sqrt{p_2}+\sqrt{p_3})^2+(\sqrt{p_1}+\sqrt{p_3})^2
\end{equation}
with an ADG-squeezed channel as the identity channel.
\end{proposition}

\begin{figure*}[t]
\centering
    \begin{minipage}[t]{0.49\textwidth}
    \centering
    \includegraphics[width=0.9\linewidth]{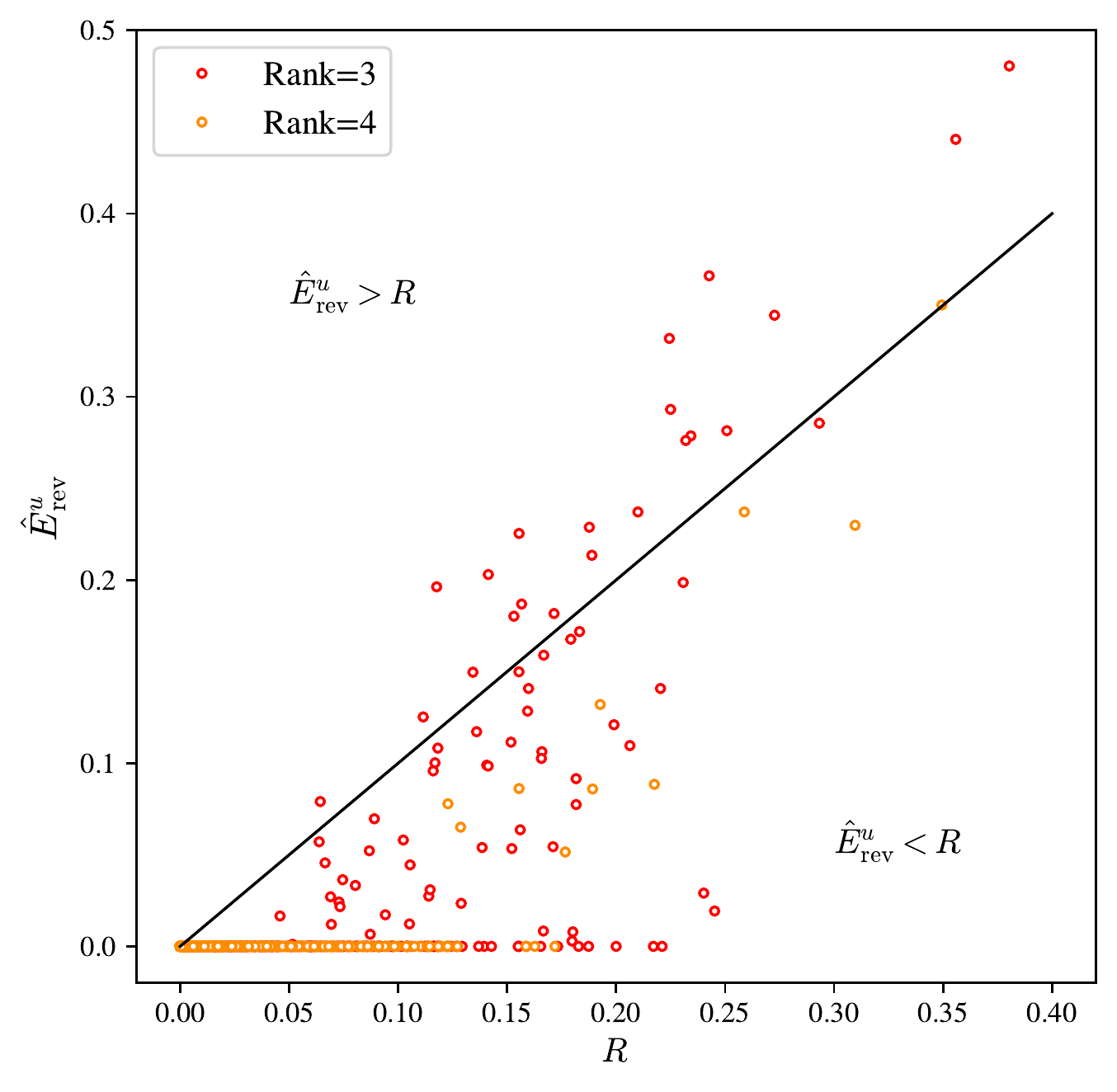}
    \caption{Upper bounds on one-way distillable entanglement of random qubit states. The black line indicates states $\rho_{AB}$ for which $\hat{E}_{\rm rev}^{u}(\rho_{AB})=R(\rho_{AB})$ and each dot represents a two-qubit quantum state. The dots above the black line (resp. below) indicate the states for which $\hat{E}_{\rm rev}^{u}(\rho_{AB})> R(\rho_{AB})$ (resp. $\hat{E}_{\rm rev}^{u}(\rho_{AB})< R(\rho_{AB})$). }
    \label{fig:qubit_random}
    \end{minipage}
    \hfill
    \begin{minipage}[t]{0.49\textwidth}
    \centering
    \includegraphics[width=0.9\linewidth]{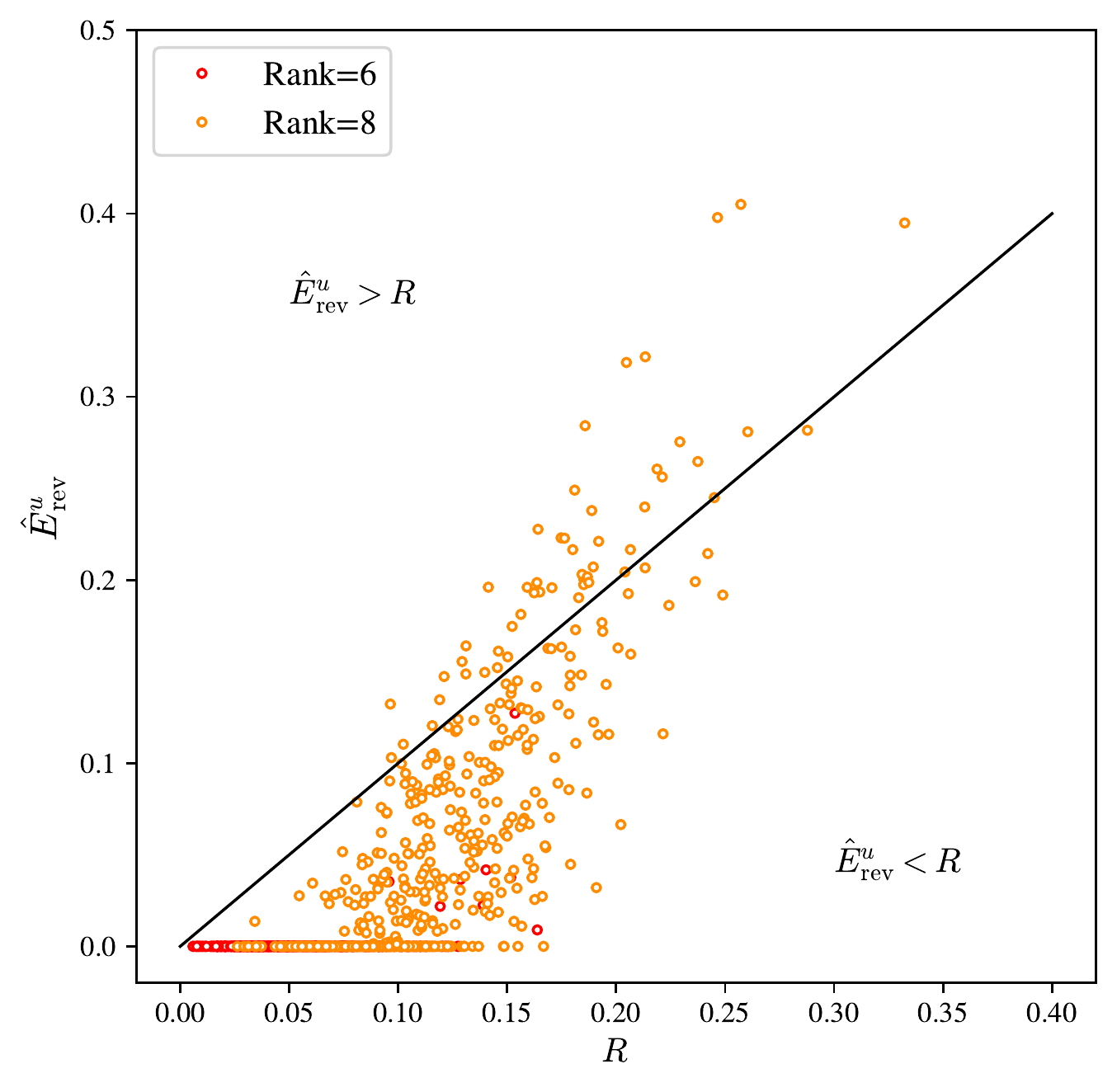}
    \caption{Upper bounds on one-way distillable entanglement of random qutrit states. The black line indicates states $\rho_{AB}$ for which $\hat{E}_{\rm rev}^{u}(\rho_{AB})=R(\rho_{AB})$ and each dot represents a two-qutrit quantum state. The dots above the black line (resp. below) indicate the states for which $\hat{E}_{\rm rev}^{u}(\rho_{AB})> R(\rho_{AB})$ (resp. $\hat{E}_{\rm rev}^{u}(\rho_{AB})< R(\rho_{AB})$). }
    \label{fig:qutrit_random}
    \end{minipage}
\end{figure*}

\begin{proof}
We first will prove 
\begin{equation*}
    2^{-\widetilde{\cR}_{\max,\adg}(\Lambda)} \geq (\sqrt{p_1}+\sqrt{p_2})^2 + (\sqrt{p_2}+\sqrt{p_3})^2+(\sqrt{p_1}+\sqrt{p_3})^2,
\end{equation*}
by using the SDP in Eq.~\eqref{sdp:chan_primal}. We show that
$\Gamma_{AB}^{\cS} = \frac{\alpha}{2}\left(\ketbra{0}{0} + \ketbra{0}{1} + \ketbra{1}{0} + \ketbra{1}{1}\right)$, is a feasible solution where $\alpha = 1-[(\sqrt{p_1}+\sqrt{p_2})^2 + (\sqrt{p_2}+\sqrt{p_3})^2+(\sqrt{p_1}+\sqrt{p_3})^2] $.
Note that the Choi state of the Pauli channel is
\begin{equation*}
    J_{AB}^\Lambda = \frac{1}{2}\left(
    \begin{array}{cccc}
    p_0 + p_3 & 0 & 0 & p_0-p_3\\
    0 & p_1+p_2 & p_1-p_2 & 0\\
    0 & p_1-p_2 & p_1+p_2 & 0\\
    p_0 - p_3 & 0 & 0 & p_0+p_3\\
    \end{array}\right), 
\end{equation*}
and the unnormalized state $\Gamma_{AB}^{\cN'} = J_{AB}^\Lambda - \Gamma_{AB}^\cS$ is
\begin{equation*}\label{eq:unnor_gamma_AB}
    \Gamma_{AB}^{\cN'} = \frac{1}{2}\left(
    \begin{array}{cccc}
    p_0 + p_3 - \alpha & 0 & 0 & p_0-p_3- \alpha\\
    0 & p_1+p_2 & p_1-p_2 & 0 \\
    0 & p_1-p_2 & p_1+p_2 & 0 \\
    p_0 - p_3 - \alpha & 0 & 0 & p_0+p_3- \alpha\\
    \end{array}\right).
\end{equation*} 
Recalling that $p_0 + p_1 + p_2 + p_3 = 1$, it is then straightforward to check that $\tr_B[\Gamma_{AB}^{\cN'}] = \tr[\Gamma_{AB}^{\cN'}]/d_A \cdot I_A$. The constraint in Eq.~\eqref{Eq:choi_extendible} corresponds to the anti-degradable condition of $\Gamma_{AB}^{\cN'}$. By direct calculation, we have 
\begin{equation*}
    \begin{aligned}
        \tr[(\Gamma_{B}^{\cN'})^2] &= \frac{1}{2}(1-\alpha)^2, \quad \det(\Gamma_{AB}^{\cN'}) = p_1 p_2 p_3(p_0-\alpha),\\ 
        \tr[(\Gamma_{AB}^{\cN'})^2] &= \alpha^2 - 2\alpha p_0 + p_0^2 + p_1^2 + p_2^2 + p_3^2.
    \end{aligned}
\end{equation*}
Then Eq.~\eqref{Eq:checkadg} holds and $\Gamma_{AB}^{\cN'}$ is anti-degradable by Lemma~\ref{lem:check_adg}, which satisfy the constraint in Eq.~\eqref{Eq:choi_extendible}.  Thus, we have proven that $\Gamma_{AB}^\cS$ is a feasible solution to the primal SDP,  which yields $2^{-\widetilde{\cR}_{\max,\adg}(\Lambda)} \geq (\sqrt{p_1}+\sqrt{p_2})^2 + (\sqrt{p_2}+\sqrt{p_3})^2+(\sqrt{p_1}+\sqrt{p_3})^2$.

Second, we will use the dual SDP in Eq.~\eqref{equ:dual_sdp_channel} to prove 
\begin{equation*}
    2^{-\widetilde{\cR}_{\max,\adg}(\Lambda)} \leq (\sqrt{p_1}+\sqrt{p_2})^2 + (\sqrt{p_2}+\sqrt{p_3})^2+(\sqrt{p_1}+\sqrt{p_3})^2.
\end{equation*}
We show that $\{M_{AB}, N_{AB}, K_{AB}, R_A\}$ is a feasible solution to the dual problem, where
\begin{equation*}\label{Eq:opt_dual_matri2}
\begin{aligned}
    &M_{AB} = \left(
    \begin{array}{cccc}
    \eta & 0 & 0 & -\eta + 1\\
    0 & \xi & \zeta & 0\\
    0 & \zeta & \xi & 0\\
    -\eta + 1 & 0 & 0 & \eta\\
    \end{array}\right),\\
    &N_{AB} = K_{AB} = -\frac{1}{2}M_{AB},
    R_A = 0,\\
    &\eta = -\frac{\sqrt{p_1} + \sqrt{p_2}}{2\sqrt{p_3}}, \xi= -\frac{\sqrt{p_1}+\sqrt{p_3}}{2\sqrt{p_2}}-\frac{\sqrt{p_2}+\sqrt{p_3}}{2\sqrt{p_1}}-1, \\
    &\zeta=\frac{\sqrt{p_1}+\sqrt{p_3}}{2\sqrt{p_2}}-\frac{\sqrt{p_2}+\sqrt{p_3}}{2\sqrt{p_1}}.
\end{aligned}
\end{equation*}
It is easy to check that when $p_0\geq p_i>0 (i=1,2,3)$, we have $M_{AB} + N_{AB} + K_{AB}= 0$, 
\begin{equation*}
\begin{aligned}
&M_{AB} \leq (1-\tr R_A/d_A)I_{AB} + R_A \otimes I_B \Leftrightarrow I_{AB} - M_{AB} \geq 0,\\
&N_{AB} \otimes I_E + P_{BE}(K_{AB} \otimes I_E)P_{BE}^\dagger\\
&\quad =\frac{1}{2}M_{AB} \otimes I_E + P_{BE}(\frac{1}{2}M_{AB} \otimes I_E)P_{BE}^\dagger \geq 0.
\end{aligned}
\end{equation*}
It also satisfies $\tr[M_{AB}J_{AB}^\Lambda] = 1-[(\sqrt{p_1}+\sqrt{p_2})^2 + (\sqrt{p_2}+\sqrt{p_3})^2+(\sqrt{p_1}+\sqrt{p_3})^2]$. Thus we have proven that $\{M_{AB},N_{AB},K_{AB},R_A\}$ is a feasible solution to the dual SDP in Eq.~\eqref{equ:dual_sdp_channel}, which yields $2^{-\widetilde{\cR}_{\max,\adg}(\Lambda)} \leq (\sqrt{p_1}+\sqrt{p_2})^2 + (\sqrt{p_2}+\sqrt{p_3})^2+(\sqrt{p_1}+\sqrt{p_3})^2$.

Thus, we arrive at $2^{-\widetilde{\cR}_{\max,\adg}(\Lambda)} = (\sqrt{p_1}+\sqrt{p_2})^2 +(\sqrt{p_2}+\sqrt{p_3})^2+(\sqrt{p_1}+\sqrt{p_3})^2$.
Since $\Gamma_{AB}^{\cS}$ is the Bell state after normalization, we know the ADG-squeezed channel is the identity channel.

\end{proof}

\section{More example for states and channels}\label{appendix:more_examples}

\begin{figure*}[t]
    \centering
    \includegraphics[width=\linewidth]{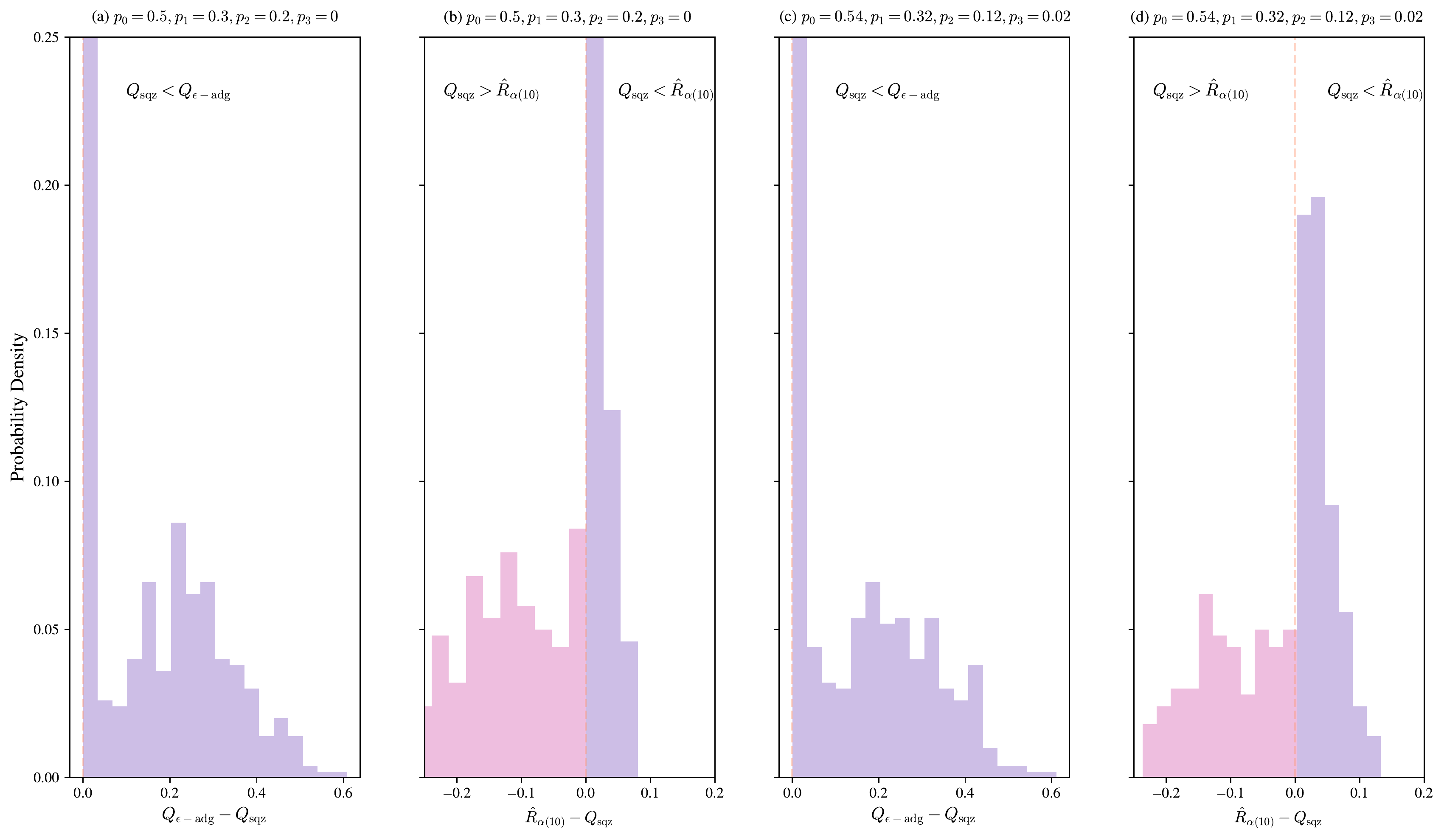}
    \caption{Upper bounds on the quantum capacity of the random mixed unitary channels. In Panel (a) and (b), we compare our bound $Q_{\rm sqz}$ in Proposition~\ref{prop:qubit_chan_sqzbound} with other upper bounds on the quantum capacity of 1000 randomly generated mixed unitary channels with fixed parameters $p_0 = 0.5$, $p_1 = 0.3$, $p_2 = 0.2$, $p_3 = 0$. The $x$-axis represents the distance between two bounds, and the $y$-axis represents the distribution of these channels. $Q_{\epsilon-\rm adg}$ is the continuity bound in Theorem~\ref{thm:conti_sutter} regarding the anti-degradability. $\hat{R}_{\alpha(10)}$ is the bound in~\cite{Fang2019a}. Panel (c) and (d) depict the same comparison for another 1000 randomly generated mixed unitary channels with fixed parameters $p_0 = 0.54$, $p_1 = 0.32$, $p_2 = 0.12$, $p_3 = 0.02$. In (a) and (c), we see that our bound always outperforms $Q_{\epsilon-\rm adg}$. In (b) and (d), we can see for many cases, our bound is tighter than $\hat{R}_{\alpha(10)}$.}
    \label{fig:mixU_appendix}
\end{figure*}

We further test the quality of our upper bound on the one-way distillable entanglement for random quantum states with different ranks. We generate 500 random bipartite states according to the Hilbert-Schmidt measure~\cite{_yczkowski_2011} with different ranks and compare $\hat{E}_{\rm rev}^{u}(\cdot)$ with $R(\cdot)$ for these states in Fig~\ref{fig:qubit_random} and Fig~\ref{fig:qutrit_random}. We can observe that our bound will be tighter than the Rains bound in many cases, both in a two-qubit case and a two-qutrit case.

We note that our upper bound for the quantum capacity can only be efficiently computed for a qubit-to-qubit channel, due to the difficulty of characterizing the extreme points of quantum channels with an input dimension larger than two~\cite{ruskai2007open}.
To further assess our bound, we present additional examples akin to those in subsection~\ref{sec:qc_pauli_chan}. We compare the performance of our method with the continuity bound in Theorem~\ref{thm:conti_sutter} and the bound $\hat{R}_{\alpha}$~\cite{Fang2019a} for more mixed unitary channels $\cU_{A\rightarrow B}(\rho) = \sum_{i=0}^k p_i U_i \rho U_i^\dagger$. Choose some fixed set of parameters and sample 1000 channels with randomly generated unitaries according to the Haar measure. The results are depicted in Fig.~\ref{fig:mixU_appendix}. We compute the distance between $Q_{\rm sqz}$ and other bounds, then obtain statistics on the distribution of these channels according to the distance value. The purple region corresponds to the cases $Q_{\rm sqz}$ is tighter, and the pink region corresponds to the cases $Q_{\rm sqz}$ is looser. We note that for these general qubit-to-qubit channels, $\hat{E}_{\rm rev}^{u}(\rho_{AB})$ and $\hat{R}_{\alpha(10)}$ apply in different cases.

\end{document}

%% file: pretex.tex
\usepackage{tikz}

\usepackage{theorem}

\newtheorem{definition}{Definition}
\newtheorem{proposition}{Proposition}
\newtheorem{lemma}[proposition]{Lemma}

\newtheorem{theorem}[proposition]{Theorem}

\newtheorem{corollary}[proposition]{Corollary}

\newenvironment{proof}{\noindent \textbf{{Proof~} }}{\hfill $\blacksquare$}

\newcounter{remark}
\newenvironment{remark}[1][]{\refstepcounter{remark}\par\medskip\noindent
\textbf{Remark~\theremark #1} }{\medskip}

\newcommand{\nc}{\newcommand}
\nc{\bra}[1]{\langle#1|}
\nc{\ket}[1]{|#1\rangle}
\nc{\ketbra}[2]{|#1\rangle\!\langle#2|}
\nc{\braket}[2]{\langle#1|#2\rangle}

\nc{\proj}[1]{| #1\rangle\!\langle #1 |}
\nc{\avg}[1]{\langle#1\rangle}
\nc{\rank}{\operatorname{Rank}}
\nc{\smfrac}[2]{\mbox{$\frac{#1}{#2}$}}
\nc{\tr}{\operatorname{Tr}}
\nc{\ox}{\otimes}
\nc{\dg}{\dagger}
\nc{\dn}{\downarrow}
\nc{\cA}{{\cal A}}
\nc{\cB}{{\cal B}}
\nc{\cC}{{\cal C}}
\nc{\cD}{{\cal D}}
\nc{\cE}{{\cal E}}
\nc{\cF}{{\cal F}}
\nc{\cG}{{\cal G}}
\nc{\cH}{{\cal H}}
\nc{\cI}{{\cal I}}
\nc{\cJ}{{\cal J}}
\nc{\cK}{{\cal K}}
\nc{\cL}{{\cal L}}
\nc{\cM}{{\cal M}}
\nc{\cN}{{\cal N}}
\nc{\cO}{{\cal O}}
\nc{\cP}{{\cal P}}
\nc{\cQ}{{\cal Q}}
\nc{\cR}{{\cal R}}
\nc{\cS}{{\cal S}}
\nc{\cT}{{\cal T}}
\nc{\cU}{{\cal U}}
\nc{\cV}{{\cal V}}
\nc{\cX}{{\cal X}}
\nc{\cY}{{\cal Y}}
\nc{\cZ}{{\cal Z}}
\nc{\cW}{{\cal W}}
\nc{\adg}{{\rm ADG}}
\nc{\csupp}{{\operatorname{csupp}}}
\nc{\qsupp}{{\operatorname{qsupp}}}
\nc{\var}{{\operatorname{var}}}
\nc{\rar}{\rightarrow}
\nc{\lrar}{\longrightarrow}
\nc{\polylog}{{\operatorname{polylog}}}
\nc{\wt}{{\operatorname{wt}}}
\nc{\supp}{{\operatorname{supp}}}

\nc{\argmin}{{\operatorname{argmin}}}

\def\x{\xi}

\nc{\RR}{{{\mathbb R}}}
\nc{\CC}{{{\mathbb C}}}
\nc{\FF}{{{\mathbb F}}}
\nc{\NN}{{{\mathbb N}}}
\nc{\ZZ}{{{\mathbb Z}}}
\nc{\PP}{{{\mathbb P}}}
\nc{\QQ}{{{\mathbb Q}}}
\nc{\UU}{{{\mathbb U}}}
\nc{\EE}{{{\mathbb E}}}
\nc{\id}{{\operatorname{id}}}

\nc{\CHSH}{{\operatorname{CHSH}}}

\nc{\rU}{\mbox{U}}

\nc{\ob}[1]{#1}

\nc{\SEP}{{\text{\rm SEP}}}
\nc{\NS}{{\text{\rm NS}}}
\nc{\LOCC}{{\text{\rm LOCC}}}
\nc{\PPT}{{\text{\rm PPT}}}
\nc{\EXT}{{\text{\rm EXT}}}
\nc{\Sym}{{\operatorname{Sym}}}

\nc{\ERLO}{{E_{\text{r,LO}}}}
\nc{\ERLOCC}{{E_{\text{r,LOCC}}}}
\nc{\ERPPT}{{E_{\text{r,PPT}}}}
\nc{\ERLOCCinfty}{{E^{\infty}_{\text{r,LOCC}}}}
\nc{\Aram}{{\operatorname{\sf A}}}

\newcommand{\eps}{\varepsilon}

\makeatletter
\def\grd@save@target#1{%
  \def\grd@target{#1}}
\def\grd@save@start#1{%
  \def\grd@start{#1}}
\tikzset{
  grid with coordinates/.style={
    to path={%
      \pgfextra{%
        \edef\grd@@target{(\tikztotarget)}%
        \tikz@scan@one@point\grd@save@target\grd@@target\relax
        \edef\grd@@start{(\tikztostart)}%
        \tikz@scan@one@point\grd@save@start\grd@@start\relax
        \draw[minor help lines,magenta] (\tikztostart) grid (\tikztotarget);
        \draw[major help lines] (\tikztostart) grid (\tikztotarget);
        \grd@start
        \pgfmathsetmacro{\grd@xa}{\the\pgf@x/1cm}
        \pgfmathsetmacro{\grd@ya}{\the\pgf@y/1cm}
        \grd@target
        \pgfmathsetmacro{\grd@xb}{\the\pgf@x/1cm}
        \pgfmathsetmacro{\grd@yb}{\the\pgf@y/1cm}
        \pgfmathsetmacro{\grd@xc}{\grd@xa + \pgfkeysvalueof{/tikz/grid with coordinates/major step}}
        \pgfmathsetmacro{\grd@yc}{\grd@ya + \pgfkeysvalueof{/tikz/grid with coordinates/major step}}
        \foreach \x in {\grd@xa,\grd@xc,...,\grd@xb}
        \node[anchor=north] at (\x,\grd@ya) {\pgfmathprintnumber{\x}};
        \foreach \y in {\grd@ya,\grd@yc,...,\grd@yb}
        \node[anchor=east] at (\grd@xa,\y) {\pgfmathprintnumber{\y}};
      }
    }
  },
  minor help lines/.style={
    help lines,
    step=\pgfkeysvalueof{/tikz/grid with coordinates/minor step}
  },
  major help lines/.style={
    help lines,
    line width=\pgfkeysvalueof{/tikz/grid with coordinates/major line width},
    step=\pgfkeysvalueof{/tikz/grid with coordinates/major step}
  },
  grid with coordinates/.cd,
  minor step/.initial=.2,
  major step/.initial=1,
  major line width/.initial=2pt,
}
\makeatother
